\documentclass{vldb}
\newcommand{\shane}[1]{\textrm{\textcolor{orange}{Shane says: #1}}}
\newcommand{\bao}[1]{\textrm{\textcolor{blue}{Bao says: #1}}}
\newcommand{\sheng}[1]{{\color{red}Sheng says: #1}}%

\newcommand{\myparagraph}[1]{\vspace{0.6\baselineskip}\noindent{\textbf{#1.}}~}
\newcommand{\var}[1]{\mbox{\emph{#1}}}
\newcommand{\svar}[1]{\mbox{\scriptsize\emph{#1}}}

\newcommand{\avar}[1]{\mbox{#1}}

\usepackage{siunitx}
\usepackage{makecell}
\usepackage{graphicx}
\usepackage{enumerate}
\usepackage[linesnumbered,ruled,lined,boxed,noend]{algorithm2e}
\usepackage{epstopdf}
\usepackage[normalem]{ulem}
\usepackage{soul}
\usepackage{xcolor}
\usepackage{multirow}
\usepackage[T1]{fontenc}
\usepackage{url}
\usepackage{color}
\usepackage{amsmath}
\usepackage{enumitem}
\usepackage{verbatim}
\usepackage{tabularx}
\usepackage{colortbl}
\usepackage{threeparttable}
\usepackage{booktabs}
\usepackage{xspace}
\usepackage{times}
\usepackage{shortvrb}
\usepackage{mathtools}
\usepackage[mathscr]{eucal}
\usepackage{tikz}
\usetikzlibrary{positioning,calc,intersections}
\usepackage{times}
\usepackage{url}
\usepackage{balance}
\usepackage{flafter}
\usepackage{lineno}
\usepackage{pgf}
\usepackage{mathrsfs}
\usepackage{algpseudocode}
\usetikzlibrary{arrows}
\usepackage{xcolor}
\usepackage[FIGTOPCAP]{subfigure}

\SetCommentSty{mycommfont}

\newtheorem{example}{Example}

\newtheorem{definition}{Definition}

\newtheorem{lemma}{Lemma}

\SetKwComment{Comment}{$\triangleright$\ }{}

\newcommand{\rknnt}{\textbf{R$k$NNT}}
\newcommand{\knn}{\textbf{$k$NN}}

\newcommand{\rknn}{\textbf{R$k$NN}}

\newcommand{\route}{\mathcal{R}}
\newcommand{\tran}{\mathcal{T}}
\newcommand{\rs}{\mathcal{D}_{\route}}
\newcommand{\ts}{\mathcal{D}_{\tran}}
\newcommand{\tdist}[2]{\var{dist}(#1,#2)}
\newcommand{\pdist}[2]{\var{distance}(#1,#2)}
\newcommand{\hs}[2]{H_{#1:#2}} 
\newcommand{\sfill}{\mathcal{S}_{\svar{filter}}}
\newcommand{\scnd}{\mathcal{S}_{\svar{cnd}}}

\newcommand{\fr}{{\textbf{RR-tree}}}
\newcommand{\sr}{{\textbf{TR-tree}}}

\newcommand{\orgin}{v_s}
\newcommand{\destination}{v_e}
\newcommand{\rknntweight}[1]{\omega(#1)}
\newcommand{\maxrknnt}{\textbf{MaxR$k$NNT}}
\newcommand{\minrknnt}{\textbf{MinR$k$NNT}}
\newcommand{\traveldistance}[1]{\psi(#1)}
\newcommand{\planroute}{R}
\newcommand{\currentroute}{R^*}
\newcommand{\srefine}{\mathcal{S}_{\svar{refine}}}

\newcommand{\queue}{\mathcal{Q}}
\newcommand{\graph}{\mathcal{G}}
\newcommand{\sresult}{\mathcal{S}_{\svar{result}}}

\newcommand{\md}[2]{\mathcal{M}_\psi[#1][#2]}

\newcommand{\compareone}{\textbf{Filter-Refine}}
\newcommand{\comparetwo}{\textbf{Voronoi}}
\newcommand{\comparethree}{\textbf{Divide-Conquer}}

\newcommand{\maxcompareone}{\textbf{Bruteforce}}

\newcommand{\rp}{r}
\newcommand{\tp}{t}
\newcommand{\crs}[1]{\mathcal{C}(#1)}
\newcommand{\scomment}[1]{\iffalse#1\fi}
\newcommand{\plist}{\textbf{PList}}
\newcommand{\nlist}{\textbf{NList}}

\begin{document}
\sloppy
\title{Reverse {\huge $k$} Nearest Neighbor Search over Trajectories}

\author{
	\begin{tabular}[t]{c@{\extracolsep{8em}}} 
		Sheng Wang$^{\ast}$~~~~~Zhifeng Bao$^{\ast}$~~~~~
		J. Shane Culpepper$^{\ast}$~~~~~Timos Sellis$^{\dagger}$~~~~~Gao Cong$^{\ddagger}$ \\
		$^{\ast}$RMIT University,
		$^{\dagger}$Swinburne University of Technology, $^{\ddagger}${Nanyang Technological University}\\
		$^{\ast}${firstname.surname@rmit.edu.au} ~~~~~$^{\dagger}$
		{tsellis@swin.edu.au}~~~~~$^{\ddagger}${gaocong@ntu.edu.sg}
	\end{tabular}
}
\maketitle


\begin{abstract}
	GPS enables mobile devices to continuously provide new opportunities
	to improve our daily lives.
	For example, the data collected in applications created by Uber
	or Public Transport Authorities can be used to plan transportation
	routes, estimate capacities, and proactively
	identify low coverage areas.
	In this paper, we study a new kind of query -- {\em Reverse $k$
		Nearest Neighbor Search over Trajectories} (\rknnt), which can be
	used for route planning and capacity estimation.
	Given a set of existing routes $\rs$, a set of passenger transitions $\ts$, and a
	query route $Q$, an \rknnt~query returns all transitions that take
	$Q$ as one of its $k$ nearest travel routes.
	To solve the problem, we first develop an index to handle dynamic
	trajectory updates, so that the most up-to-date transition data
	are available for answering an {\rknnt} query.
	Then we introduce a filter refinement framework for processing
	{\rknnt} queries using the proposed indexes.
	Next, we show how to use {\rknnt} to solve the optimal route planning
	problem {\maxrknnt} (\minrknnt), which is to search for the
	optimal route from a start location to an end location that
	could attract the maximum (or minimum) number of passengers based on a
	pre-defined travel distance threshold.
	Experiments on real datasets demonstrate the efficiency and
	scalability of our approaches.
	To the best of our best knowledge, this is the first work to study
	the {\rknnt} problem for route planning.
\end{abstract}

\section{Introduction}
\scomment{
	1. problem: describe traffic problem directly (demand), data model: transition:
	two points, dynamic, route: multiple points,
	the core operator and my problem, and rknn, then rknnt, try to find transition to take my route to commute, features: real-time, significance and applying area. Challenge: query is multiple points, route is also a challenge
	then related work: our features: 1)no training time and the dataset is static, 2) basic operator. classify.
}

{\em Reverse $k$ Nearest Neighbor} ({\rknn})
queries have attracted considerable attention
{\cite{stanoi2000reverse,tao2004reverse,wu2008finch,Cheema2011,
		Yang2014,Yang2015a}.
	An {\rknn} query aims to find all the points among a set of points that
	take a query point as their $k$ nearest neighbors. The {\rknn} query
	has many applications in resource allocation, decision support
	and profile-based marketing, etc.
	For example, {\rknn}~can be applied to estimate the number of
	customers of a planned restaurant among all existing restaurants.
	
	The increasing prevalence of GPS-enabled devices provides new
	opportunities to obtain many real trajectory data in a short period of
	time, such as the GPS trajectories of taxi or uber drivers, check-in trajectories of social media users in Foursquare, which describe the movements/transitions of people.
	In this paper, we will explore the {\rknn}~search over multiple-point
	trajectories (which is referred to as {\rknnt}).
	In a nutshell,
	the {\rknnt} query~can
	be described as: taking a planned (or
	existing) route as a query $Q$, return
	all the passengers who will take the query route $Q$ as one of
	the $k$ nearest routes among the route set $\rs$ to
	travel. Here, a passenger's movement is modeled as a combination of an
	origin and a destination~{\cite{Liu2016}} such as home and office,
	which is called a {\em transition}.
	As shown in Figure~\ref{fig:motivation}, two red points represent a
	single transition $\tran$ of a passenger.
	Among the two available routes, $\tran$ would favour Route 2 as both
	points in $\tran$ can find a point closer in Route 2 than in Route 1.

	The main difference between our problem and {\rknn} is that our query
	is a route, and our data collections contain routes and user transitions
	respectively.
	{\rknnt}~can estimate the passengers that will take the query
	route to travel.
	Another significant difference of our {\rknnt} from previous work is that
	the transition data is very dynamic and new transitions will arrive continuously, such as the request of uber passengers.
	Therefore, it is important to take this into consideration in the solution to answering the {\rknnt} query.
	

	
	\scomment{
		Recently, Liu et al.~{\cite{Liu2016}} proposed a choice model to
		determine whether a passenger would take a route based on the user's
		historical movement data.
	}

	\begin{figure}[h]
		\vspace{-1.5em}
		\centering
		\includegraphics[height=2.5cm]{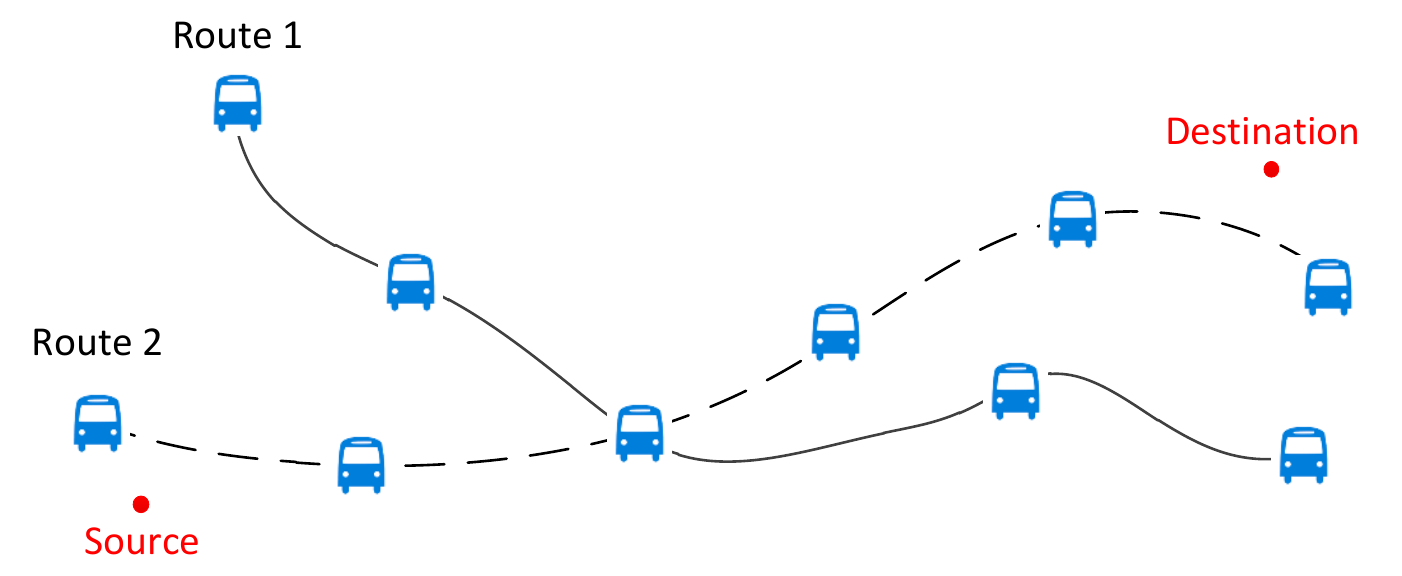}
		\vspace{-1em}
		\caption{Bus Routes and User Transition.}
		\vspace{-1em}
		\label{fig:motivation}
	\end{figure}
	
	{\rknnt} serves as a core yet fundamental operator in many applications in transportation field. The first and foremost one is to estimate the capacity of a route based on passengers' dynamic movements. On top of capacity estimation, {\rknnt} can be used for the following promising applications. \\
	{\textbf{1) Optimal Route Planning.}}
	(i) Among a set of candidate routes, an {\rknnt} can be used to find the optimal
	route which has the maximum (minimum) number of passengers among a
	set of candidate routes.
	For Uber drivers, finding a route with the maximum number of
	passengers can increase profitability (the uber fare will be increased with a surge of passenger requests) and the chance of being hired.
	For ambulance and fire truck drivers, finding a route which
	has the fewest people can increase response times in emergency
	situations. We will address problem (i) in this paper as well.
	(ii) Furthermore, by taking the temporal factor into consideration, i.e., user transitions at different time periods, it can help further adjust the frequency of planned vehicles on the planned routes, in order to save running cost for either individual vehicle drivers or public transportation authority~\cite{Ceder2007}.
	%
	
	The first challenge in answering {\rknnt} lies in how to
	prune the transitions which cannot be in the results without
	accessing every transition.
	A straightforward method is to conduct a {\knn} search for every
	transition, and then check the resulting ranked lists to see whether the
	query is a \knn.
	This method is intractable when there are a large number of transitions
	and new transitions are being added to the database.
	For {\maxrknnt}, a brute force method can be used to find all
	candidate routes whose travel distances do not exceed the distance
	threshold, and then an {\rknnt} search can be executed on each
	candidate, and finally the one with maximum number of passengers
	is selected as the answer.
	This method is shown to be inefficient in our experimental study.
	Similar to {\rknnt}, it is crucial to prune out candidates which
	cannot be an optimal route.
	Another challenge is how to support \emph{dynamic updates} as old transitions expire and new transitions arrive.
	
	To addresses this challenge, we first build two R-tree
	indexes $\fr$ and $\sr$ combined with two inverted indexes {\plist}
	and {\nlist} for the route and transition sets, respectively.
	Then, we choose a set of route points from the existing routes to form a
	{\em filtering set} by traversing the route index $\fr$.
	By drawing bisectors between the route points in the filtering set and
	the query, an area can be found where any transition point inside
	can not have the query as a nearest neighbor.
	After finding the {\em filtering set}, the pruning of transitions
	starts by traversing $\sr$ and checking if a node can be
	filtered by more than $k$ routes in the filtering set.
	Finally, all of the candidate transitions are verified using the
	filtered nodes during the traversal of $\fr$.
	
	Next, we explore the optimal route planning problem, where we consider a graph formed by the bus network.
	Given a starting location and a destination, we find the optimal route
	$\route$ which connects the two locations in a bus network, and
	maximizes (minimize) the number of passengers that take $\route$ as
	its $\knn$ without exceeding a distance threshold.
	We call this problem as the {\maxrknnt} (\minrknnt).\\
	{\textbf{2) Bus Advertisement Recommendation.}}
	As an {\rknnt} query for a route can locate the set of passenger transitions
	who would take the travel route, and we can obtain the profiles of
	potential passengers using social networks, a
	deeper analysis of the common interests among passengers who take
	similar travel routes can be found.
	Consequently, we can select and broadcast advertisements that
	will have the largest influence on passengers taking a route.

	For {\maxrknnt}, a weighted graph is built by the pre-computed
	{\rknnt} set for every vertex.
	Then we start from the first vertex and access the neighbor vertex
	$v$ to compute a partial route $R$.
	Then a reachability check on $R$ is performed to see whether the
	estimated lower bound travel distance of $R$ is greater than the
	distance threshold.
	Next, the dominance table of $v$ is checked to see if $R$ can
	dominate other partial routes which terminate at $v$.
	Further checks on the route are made when $R$ meets all the
	conditions.
	
	In summary, the main contributions of this paper are:
	\begin{itemize}[leftmargin=1em]
		\itemsep 0em
		\item We investigate the {\rknnt} problem for the first time, which
		serve as a core yet frequently adopted operator in many practical applications. In particular, we explore how to use
		{\rknnt} to plan an optimal route which
		attracts the most (or fewest) passengers (Section~\ref{sec:problem}).
		\item We propose a filtering-refinement framework which can prune routes
		using a {\em filtering set} (in Section~\ref{sec:baseline}) and a
		voronoi-based optimization to further improve the efficiency
		(Section~\ref{sec:optimization}).
		\item We also introduce {\maxrknnt} (\minrknnt) queries which can be
		used to find the optimal route that attracts the maximum (minimum)
		number of passengers in a bus network (Section~\ref{sec:maxrknnnt}).
		\item We validate the practicality of our approaches using real world
		datasets (Section~\ref{sec:experiment}).
	\end{itemize}
	%

	
	

\begin{figure}[htp]
	\vspace{-1em}
	\centering
	\includegraphics[height=2cm]{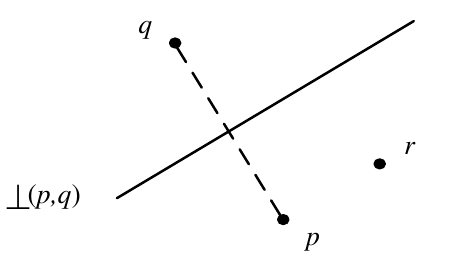}
	\vspace{-1em}
	\caption{Pruning by half-space.}
	\label{fig:tpl}
	\vspace{-2em}
\end{figure}

\section{Preliminary \& Related Work}
\label{sec:method}
In this section, we first compare the difference between our work and
classic {\rknn} search over point and moving object data.
Further, we will review related work on route planning.

\subsection{{\large \rknn}}\label{sec:rnns}
\myparagraph{{\rknn} on Spatial Points} Most existing {\rknn} search
work focuses on static point
data,
and often employ a pruning-refinement frameworks to avoid scanning the
entire dataset.
However, these approaches cannot easily be translated to route search
where both queries and collections are of multi-point trajectories.
Given a set of candidates as a query, a maximizing {\rknn} query
finds the optimal trajectory which satisfies the maximum number of
results~{\cite{Wong2009,Zhou2011}}.

How to improve the search performance has attracted much attentions 
over last decades.
The basic intuition behind filtering out a point $p$ is to find
another point which is closer to $p$ than the query point $q$
~{\cite{tao2004reverse, wu2008finch,Cheema2011,Yang2014}}.
Here, we review the half space method and use a simple example to
show how pruning works.
Figure~\ref{fig:tpl} shows a query point $q$ and a data point $p$.
As we can see, a perpendicular bisector divides the whole space into
two sub-spaces, and all points inside the lower subspace would prefer 
$p$ as a nearest neighbor, such as point $r$.
For the reverse nearest neighbor search, $r$ may be filtered 
from the candidate set of query $q$.
More specifically, $r$ can be filtered out if it can be pruned by at
least $k$ such points.

\myparagraph{{\rknn} on Moving Objects} 
Given a moving object dataset $D$ and a candidate point set $O=\left(
o_1,o_2,...,o_m \right)$ as query, Shang et al.~{\cite{Shang2011Finding}} 
find the optimal point from $O$ such
that the number of moving objects that choose $o_i$ as a nearest
neighbor is maximized.
Specifically, they proposed a {\em Reverse Path Nearest Neighbor} (R-PNN)
search which finds the nearest point, and not the $k$ nearest points.
Shang et al.
introduced the concept of {\em influence-factor} to determine the
optimal point.
The influence-factor of $o$ is $f$ if $o$ is the nearest neighbor of
$f$ trajectories for all candidate points in $O$.
Cheema et al.~{\cite{Cheema2012}} proposed a continuous reverse
nearest neighbors query to monitor a moving object and find all
static points that take the moving object query as a $k$ nearest
neighbor.
Recently, Emrich et al.~{\cite{emrich2014reverse}} solved a problem
of RNN search with ``uncertain'' moving object trajectories using
a Markov model approach.
A moving object is treated as a result when it always takes the query
object as a nearest neighbor for every time stamp within a given time
interval.
All of these approaches target a single point rather than a
transition of multiple-points, which is the focus of our work.

\subsection{Route Searching}
\label{sec:bnp}
\myparagraph{Bus Route Planning} Bus network design is known to be a
complex, non-linear, non-convex, multi-objective NP-hard
problem~{\cite{Chen2014a}}.
Based on existing bus networks, Pattnaik et
al.~{\cite{pattnaik1998urban}} proposed a heuristic method which uses
a genetic algorithm to minimize the cost of passengers and operators.
Yang et al.~{\cite{yang2007parallel}} used ant colony algorithms to
maximize the number of direct travelers between two nearby bus stops.
\scomment{Such work uses population estimation and user surveys to find
	a sub-optimal route, but do not consider realtime dynamic movements
	of passengers. }
Population estimation and user surveys~{\cite{Xian2013}} around the
planned route are traditional ways to estimate the number
of passengers that may use the planned travel route.
However, the data is usually out-of-date as 
census data may only be gathered every five years, and may not reflect
current travel patterns.

Chen et al.~{\cite{Chen2014a}} tried to approximate night time bus
route planning by first clustering all points in
taxi trajectories to determine ``hot spots'' which could be bus
stops, and then created a bus route graph based on the connectivity
between two stops.
Based on human mobility patterns, Liu et al.~{\cite{Liu2016}}
proposed a localized transportation choice model, which can
predict bus travel demand for different bus routes by taking
into account both bus and taxi travel demands.
However, the method must scan the static records for all of the
customers, which is inefficient in practice, and the model has to be
rebuilt whenever the records are updated.
\scomment{
	These two works all scan entire dataset to estimate the
	transportation capacity for a route, and do not support dynamic
	updates.
	
	\bao{I do not think it is closely related to our \maxrknnt problem. I suggest removing this section!}
	
	\sheng{We need to the shortest distance computation, and we need to mention the graph search}}

\myparagraph{Shortest Route Searching}
Given a starting vertex and an ending vertex, the classical route
planning problem is to find the shortest path in a graph.
Best First Search (BFS) and Depth First Search (DFS) are two commonly
used algorithms for this problem.
An extension of this problem is {\em $k$ Shortest Path searching}
($k$SP) ~{\cite{Yen1971,Aljazzar2011}}, which aims to find the $k$
shortest paths from a start vertex $s$ to a target vertex $t$ in a
directed weighted graph $G$.
Yen's algorithm~{\cite{Yen1971}} is a derivative algorithm for
ranking the $k$ shortest paths between a pair of nodes.
The algorithm always searches the shortest paths in a tree containing
the $k$ shortest loop free paths.
The shortest one is obtained first, and the second shortest path is
explored based on the previous paths.
The {\em Constraint Shortest Path} (CSP) problem
~{\cite{Lozano2013,Bolivar2014}} applies resource constraints on each
edge, and solves the shortest path search problem based on these
constraints.
An example constraint would be time costs.

\section{Problem Definition}
\label{sec:problem}
In this section, we formally define the \rknnt~problem and important
notations are recorded in Table~\ref{tab:notation}.
\begin{table}
	\centering
	\caption{Summary of Notation}
	\label{tab:notation}
	\begin{tabular}{|c|c|}
		\hline
		\textbf{Notation} & \textbf{Definition}  \\
		\hline
		\hline
		$\route$ & The route composed of points $\{\rp_1,...,\rp_n\}$\\
		\hline
		$\tran$ & The transition $\{\tp_o,\tp_d\}$\\
		\hline
		$Q$ & Query route $\{q_1,...,q_m\}$\\
		\hline
		$\rs$, $\ts$ & The route and transition sets\\
		\hline		
		$\tdist{\tp}{\route}$ & Distance from transition point $\tp$ to $\route$\\
		\hline
		$\bot(q,\rp)$ & Perpendicular bisector between $q$ and $\rp$\\
		\hline	
		\makecell{$\hs{q}{\rp}$, $\hs{\rp}{q}$} & \makecell{Two half-planes 
			divided by
			\\ perpendicular bisector $\bot(q,\rp)$} \\
		\hline
		$\hs{\rp}{Q}$, $\hs{\route}{Q}$ & Filtering space formed by $Q$ with $\rp$ and $\route$ \\
		\hline
		$\crs{\rp}$ & Crossover route set of $\rp$ \\
		\hline
		$\sfill$, $\srefine$ & Filtering set and filtered node set\\
		\hline
		$\scnd$, $\sresult$ & Transition candidates and result set\\
		\hline
		$\var{root}_r$ ($\var{root}_t)$ & Root of Route R-tree ( Transition R-tree )\\
		\hline
		$\mathcal{V}_{\route,Q}$ & Voronoi diagram formed by $\route$ and $Q$ \\
		\hline
		$\mathcal{G}$ & Weighted graph\\
		\hline
		$\tau$ & Travel distance threshold\\
		\hline
		$\rknntweight{\route}$, $\traveldistance{\route}$ & \rknnt~set and travel distance of $\route$ in $\graph$\\
		\hline
		$\mathcal{M}_\psi[i][j]$ & \makecell{Lower bound matrix of 
			$\traveldistance{\route}$ \\ 
			where $\route$ starts from vertex $i$ to $j$}\\
		\hline
	\end{tabular}
	\vspace{-1.5em}
\end{table}

\begin{definition}{\textbf{(Route)}}
	\label{def:rt}
	A route $\route$ of length $n$ is a sequence of points 
	$\left(\rp_1,\rp_2,...,\rp_n\right),n \ge 2$, 
	where $\rp_i$ is a point represented by 
	(latitude, longitude).
\end{definition}

\begin{definition}{\textbf{(Transition)}}
	\label{def:ct}
	A transition $\tran$ contains an \text{o}rigin point $\tp_o$
	and a \text{d}estination point $\tp_d$.
\end{definition}

Both a route and a transition are composed of discrete points
called {\em route point} $\rp$ and {\em transition point} $\tp$
respectively.
We use $\ts$ and $\rs$ to denote the transition set and route set.

\begin{definition}{\textbf{(Point-Route Distance)}}
	\label{def:dis}
	Given a transition point $\tp \in \tran$ and a route
	$\route$, the distance $\tdist{\tp}{\route}$ from
	$\tp$ to $\route$ is the minimum Euclidean distance from
	$\tp$ to every point of $\route$, and calculated as:
	\begin{equation}
	\label{equ:prd}
	\tdist{\tp}{\route} = \min_{\rp \in \route}\pdist{\tp}{\rp}
	\end{equation}
\end{definition}

Based on the point-route distance function, the {\knn} search of a
transition point $\tp$ is defined as: 
\begin{definition}{(\knn)}
	\label{def:knn}
	Given a set of routes $\rs$, the {\knn} search of a transition
	point $\tp \in\tran$ retrieves a set $S\in \rs$ of $k$ routes
	such that for all $\route \in S$, and for all $\route^{'} \in \rs-S$: $
	\tdist{\tp}{\route}
	\geq\tdist{\tp}{ \route^{'}}$.
\end{definition}

In particular, two types of {\knn} are supported for a transition $\tran$,
which can also be found in \cite{emrich2014reverse}.
\begin{enumerate}
	\item $\exists$\knn: $\tran$~takes $\route$ as a {\knn} iff there
	exists a point $\tp \in \tran$ taking $\route$ as {\knn}. So, 
	\textbf{$\exists$kNN}($\tran$) $= \knn(\tp_o) \cup \knn(\tp_d)$.
	\item $\forall$\knn: $\tran$~takes $\route$ as a {\knn} iff
	both points $\tp_o$ and $\tp_d$ take $\route$ as their {\knn}. So,
	\textbf{$\forall$kNN}($\tran$) =$\knn(\tp_o) \cap \knn(\tp_d)$.
\end{enumerate}

Now, we can formally define the reverse $k$ nearest neighbor query
over trajectories.
\begin{definition}{\textbf{(R$k$NNT)}}
	\label{def:rknnt}
	Given a set of routes $\rs$, a set of transitions $\ts$, and
	a query route $Q$, $\exists\rknnt(Q)$ ($\forall$\rknnt(Q))
	retrieves all transitions $ \tran \in \ts$, such that for all $\tran$:
	$Q\in\textbf{$\exists$\knn}(\tran)$
	($\textbf{$\forall$kNN}(\tran)$).
\end{definition}
\scomment{
	\shane{The last part of this definition is not clear.}
	\sheng{I polished it}.}

\begin{figure}
	\centering
	\vspace{-2em}
	\includegraphics[height=5cm]{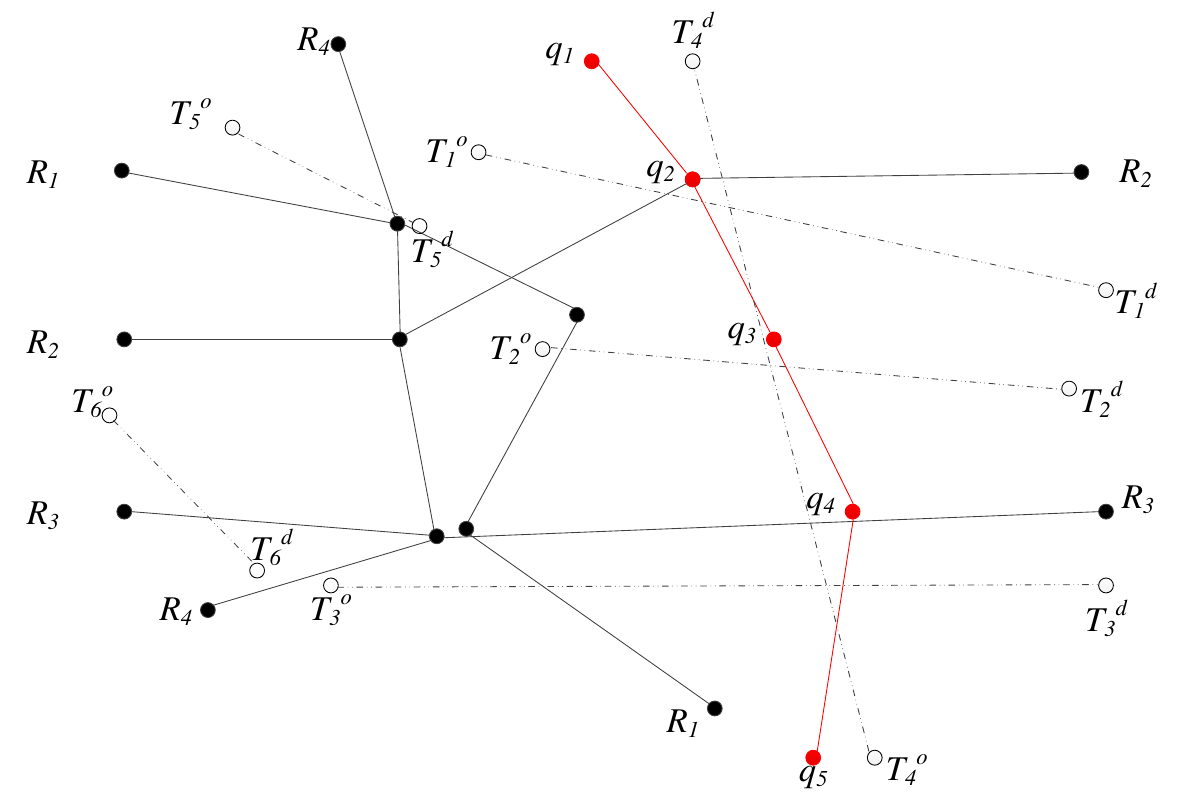}
	\vspace{-1em}
	\caption{Example of routes and transitions.}
	\label{fig:examp}
\end{figure}
\vspace{-1em}
\begin{example}
	In Figure~\ref{fig:examp}, $R_1$, $R_2$, $R_3$ and $R_4$ are {\em
		routes}.
	$T_{1}, T_{2}, T_{3}, T_{4}, T_{5},T_{6}$ are {\em transitions}, and
	$T_{1}^o$ and $T_{1}^d$ denote the origin point and destination point
	for transition $T_1$.
	The query route $Q$ is composed of $5$ query points (in red).
	If we take the $\forall$\rknnt~query, as point $T_{4}^o$ and
	$T_{4}^d$ take $Q$ as the nearest route, $T_4$ will be the result of
	$\forall\rknnt(Q)$.
\end{example}

\begin{lemma}
	\label{le:allexist}
	Given a query $Q$, $\forall\rknnt(Q) \subseteq \exists\rknnt(Q)$.
\end{lemma}
\begin{proof}
	Given a query $Q$, $\forall\rknnt(Q)$ returns a set of transitions
	where both origin and destination points have the query as a {\knn}, then
	such transitions will also belong to the result of $\exists\rknnt(Q)$,
	so $\forall\rknnt(Q) \subseteq \exists\rknnt(Q)$.
	Let $\varDelta(Q) = \exists\rknnt(Q)-\forall\rknnt(Q)$,
	$\forall \tran \in \varDelta$: $\tran$ only has one point that will
	take the query as a {\knn}, so $\varDelta(Q)
	\cap\forall\rknnt(Q)=\emptyset$.
\end{proof}

Using Lemma~\ref{le:allexist}, the set of transition points which
take $Q$ as {\knn} can be searched for first, and then $\exists\rknnt(Q)$
can be found by adding the corresponding routes.
For $\forall$\rknnt, we need to remove transitions that have
only one point in $\exists\rknnt(Q)$.
Hence, a unified framework can be proposed that answers
both $\exists$\rknnt~and $\forall${\rknnt}.
In the rest of this paper, we use {\rknnt} to represent
$\exists$\rknnt~by default for ease of composition.

\section{Capacity Estimation - a Processing Framework for {\large $\rknnt$}}
\label{sec:baseline}
In this section, we first provide a sketch of our framework to answer
the {\rknnt} query for capacity estimation, which includes the pruning
idea based on routing points, and the proposed index structures.
Then we describe each step in detail.

\begin{figure}[t]
	\centering
	\includegraphics[height=5cm]{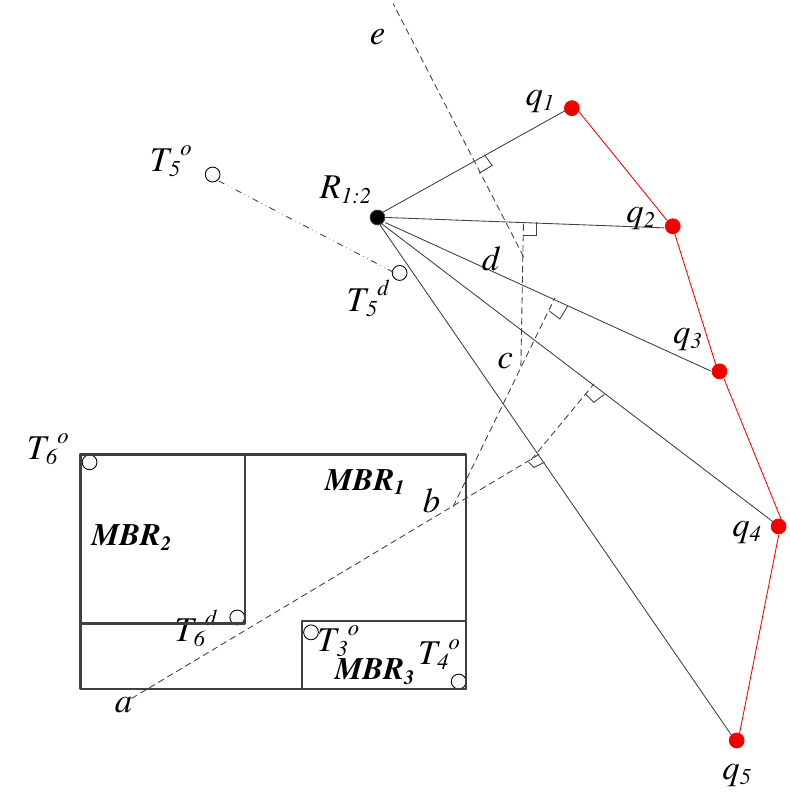}
	\vspace{-1.5em}
	\caption{Pruning by half-space for a multi-point query $Q$.}
	\vspace{-1.5em}
	\label{fig:onepoint}
\end{figure}

\subsection{Main Idea}
All impossible transition points are pruned using a $\mathsf{PruneTransition}$ algorithm, and the remaining candidates $\scnd$ are
further verified using a $\mathsf{RefineCandidates}$ algorithm to generate
the final result set $\sresult$.
Note that before pruning, a subset of routes $\sfill$ needs to be
generated for efficient pruning, the reasoning and approach 
are described in Section~\ref{sec:pc}.
In summary, the whole procedure is composed of the three steps 
in Algorithm~\ref{alg:fw}.

\begin{algorithm}
	\caption{$\mathsf{\rknnt}(Q,\var{root}_r, \var{root}_t)$}
	\label{alg:fw}
	\KwOut{$\sresult$: the result set}
	$(\sfill, \srefine) \leftarrow \mathsf{FilterRoute}(\var{root}_r, Q, k)$\tcp*{\small Sec~\ref{sec:ffs}}
	$\scnd \leftarrow \mathsf{PruneTransition}(\var{root}_t, Q,\sfill, k)$\tcp*{\small Sec~\ref{sec:tp}}
	$\sresult \leftarrow \mathsf{RefineCandidates}(Q,\scnd , \srefine)$\tcp*{\small Sec~\ref{sec:verify}}
	\Return $\sresult$\;
\end{algorithm}

\subsubsection{Pruning Characteristics}
\label{sec:pc}
By Definition~\ref{def:rknnt}, a transition takes a route as a {\knn}
if there exists at least one point (in the transition) that will take
the route as a {\knn}.
If there are more than $k$ routes which are closer to a point in a
transition than the query, then the point in this transition can be
pruned.
Such a route which helps prune transitions is called a {\em filtering
	route}.
If both points of a transition are pruned, then the transition can be
pruned safely, so pruning transitions helps to find the filtering
routes to prune the points in the transition.

\begin{lemma}
	\label{le:tr}
	If a transition point $\tp$ is 
	closer to a route point $\rp\in\route$ than $Q$,
	then $\tp$ is closer to $\route$ than $Q$.
\end{lemma}

\begin{proof}
	We have $\tdist{\tp}{\route} \le \var{distance}(\tp,\rp)$ 
	according to Equation~\ref{equ:prd}.
	If $\var{distance}(\tp,\rp)<\tdist{\tp}{Q}$, 
	then 	$\tdist{\tp}{\route} < \tdist{\tp}{Q}$.
\end{proof}

By Lemma~\ref{le:tr}, a transition point can be removed if it takes a set
of routing points from more than $k$ different routes as a {\knn}
rather than the query.
These points are called {\em filtering points}.
Next, we introduce how to prune a transition points using the filtering
point $\rp$ from the routes.

Recall the example in Figure~\ref{fig:tpl} where an {\rknn} can find
an area where the points inside the area will not take the query as the
nearest neighbor based on the half space.
Similarly, given a query $Q$, we choose a point $\rp$ from a route
$\route$ in $\rs$; then based on the straight line $\overline {\rp
	q_i}$ formed by a point $q_i$ in $Q$ to $\rp$, the perpendicular
bisector $\bot(q_i,\rp)$ is used to cut the space into two
half-planes: $\hs{\rp}{q_i}$ and $\hs{q_i}{\rp}$ which contain $\rp$
and $q_i$, respectively.
For every point $q_i$ in $Q$, there is a $\hs{\rp}{q_i}$.

The intersection of all of the half spaces forms the filtering space
$\hs{\rp}{Q}$ defined as:
\begin{definition}{(\textbf{Filtering Space})}
	\label{de:fs}
	Given a route point $\rp$ and a query $Q$, 
	the intersection of all $\hs{\rp}{q_i}$ forms a filtering space:
	\begin{equation}
	\hs{\rp}{Q} = \bigcap\limits_{{q_i} \in Q} {\hs{\rp}{q_i}} 
	\end{equation}
	the point $\rp$ which belongs to $\route \in \rs$ is called as the 
	filtering point.  \\
	\scomment{
		\bao{why we need to highlight the case of one query point?}\\
		\sheng{I moved it to Section~\ref{sec:dc}}.}
\end{definition}

As shown in Figure~\ref{fig:onepoint}, we can see that there are five
perpendicular bisectors.
They form a polyline $\overline {abcde}$ which divides the whole
space into two sub-spaces, and the left part is the filtering space
$\hs{R_{1:2}}{Q}$.
As $T_5$ is entirely located in this area, it cannot take the
query as its nearest route.
The filtering space can also help filter a set of points using spatial
indexes (see Sec.~\ref{subsec:index}).
If a maximum bounded box (MBB) such as $\var{MBR}_2$ covering
points $T_6^o$ and $T_6^d$ is located entirely inside the filtering
space, then $T_6^o$ and $T_6^d$ inside this MBB will not take the
query as a nearest neighbor and can be filtered out.

Every point in the route set $\rs$ can be a filtering point, but we
cannot choose all points in the route set $\rs$ to do pruning
especially when the whole set is large and located in external memory.
To pruning a transition point, if we access all route points every
time, the process is costly.
In Section~\ref{sec:ffs}, we introduce how to generate a subset
from the whole route set.
Overall, we can observe three key characteristics based on the above
analysis: 1) A filtering space exists between the query and a
route point;
2) If a transition point is located in more than $k$ filtering space of
query $Q$ simultaneously, then the point can be pruned; and
3) It is important to choose a subset of all routes as the filtering
set.

\subsubsection{Indexes}\label{subsec:index}
\begin{itemize}[leftmargin=1em]
	\itemsep 0em
	\item $\fr$ \& $\sr$
	are two tree indexes for point data fetched from route dataset
	$\rs$ and transition dataset $\ts$ respectively, and referred to as
	a Route R-tree ($\fr$) and a Transition R-tree ($\sr$).
	The tree indexes are created first, and every point in the leaf node of
	$\fr$ contains the IDs of the routes it belongs to.
	Every point in $\sr$ also contains the IDs of the transition it
	belongs to. Through the transition ID and route ID in the node of $\fr$ and
	$\sr$, we can get the corresponding route and transition for further
	refinement if two points of a transition are both pruned, and the two
	filtering points belong to the same route
	(See~Section~\ref{sec:verify}).
	\item $\nlist$. As we need to get all the routes that have a point inside
	a given node for verification in Section~\ref{sec:verify}.
	Hence, for each node in $\fr$, we will create a list  for every node of {\fr} by
	traversing the whole tree bottom-up to store all the IDs of routes
	inside.
	\item $\plist$.
	The inverted list of each route point is created to store the IDs of
	the corresponding routes.
	As a bus stop can be shared by many routes in a bus network, we
	call this index a {\plist}.
\end{itemize}

Our index supports dynamic updating, where new transitions and routes
can be added into the index easily.
This is in contrast the previous work~{\cite{Chen2014a,Liu2016}}
which needs to train whole dataset from scratch once there are new
data inserted.

\subsection{Key Functions}
This section describes: 1) how to 
generate the filtering set $\sfill$; 
2) how to prune and find all the candidate routes $\scnd$; and
3) how to verify the candidate routes and further refine them to find 
the final query result $\sresult$.

\subsubsection{Filtering Routes}
\label{sec:ffs}
In order to get a small filtering set $\sfill$ for a given query,
an empty filtering set $\sfill$ is initialized, and new route
points are added which cannot be pruned using the existing points of
a route in $\sfill$.
We organize all the filtered points in a point list, sorted by the
number of routes which cover each point, and denote the route set and
point set as $\sfill.R$ and $\sfill.P$ respectively, which are 
materialized using two dynamic sorted hashtables.
Specifically, for $\sfill.R$, the key is the route ID, and the values
are points of this route that cannot be filtered.
For $\sfill.P$, the key is the route point ID, and the value is a list
of routes containing the point.

Note that in a real bus network, a route point can be covered by
several routes.
If a filtering point is contained by more than $k$ routes, and a
transition takes this filtering point as the nearest neighbor rather than
the query, and then this transition point can be pruned.
We will employ this enhancement to achieve the more efficient pruning.
\begin{definition}{\textbf{(Crossover Route Set)}}
	Given a route point $\rp$, the set of routes which cover
	$\rp$ is the crossover set of $\rp$, and denoted as
	$\crs{\rp}$.
\end{definition}
For example in Figure~\ref{fig:examp}, $R_1$ and $R_4$ intersect at
the second point $R_{1:2}$, then $\crs{R_{1:2}} = \{R_1, R_4\}$.
Using the {\plist}, we can retrieve the {\em crossover route set} of a
point $\rp$ easily, where $\crs{\rp}=\var{PList}[\rp]$.
The crossover route set of each filter point $\rp \in
\sfill.P$ can be sorted by $|\crs{\rp}|$ to give higher priority
to the points which are crossed by more routes in the filtering
phase.

Starting from the root node of {\fr}, the filtering algorithm
iteratively accesses the entries of {\fr} from a heap in ascending
order of their minimum distances to the query $Q$.
The accessed points are used for filtering the search space.
If an accessed entry $e$ of index can be filtered -- $e$ is
pruned by more than $k$ routes -- it can be skipped (see
Algorithm~\ref{alg:pt}).
Otherwise, if $e$ is an intermediate or leaf node, its children are
inserted into the heap; if $e$ is a route point and cannot be
filtered, it is inserted in the filter set $\sfill$ and its
half-space is used to filter the search space.
The filtering algorithm terminates when the heap is empty.
The details can be found in Algorithm~\ref{alg:filteringpoint}.
\begin{algorithm}[h]
	\KwOut{$\sfill$: filtering set, 
		$\srefine$: filtered node set for refinement}
	\caption{$\mathsf{FilterRoute}(\var{root}_r$, $Q$, $k$)}
	\label{alg:filteringpoint}
	$\var{minheap} \leftarrow \varnothing$, $\sfill \leftarrow \varnothing$,
	$\srefine \leftarrow \varnothing$\;
	$\var{minheap.push}(\var{root}_r)$\;
	\While{$\avar{minheap.isNotEmpty()}$}
	{
		$e \leftarrow \var{minheap.pop()}$\;
		\If{$e$ is a node}
		{
			\If{$\mathsf{isFiltered}$($Q$, $\sfill$, $e$, $k$)\label{fps:isfilter}}
			{
				$\srefine.\var{add}(e) \label{alg:fil:refine}$\tcp*{\small Filtered node set}
				\textbf{continue}\;
			}
		}
		\eIf{$e$ is a point}
		{
			$\sfill.\var{add}(e, \mathcal{C}(e))$\label{filterp:add}\tcp*{\small Filtering set}
		}
		{ 
			\ForEach{child $c$ of $e$}
			{
				$\var{minheap.push}(c,\var{MinDist}(Q, c))$\;
			}
		}
	}
	\Return $(\sfill, \srefine)$\;
\end{algorithm}

The minimum distance from a child $c$ 
to the query is computed as the minimum 
distance from every query point to the node $c$:
\begin{equation}
\var{MinDist}(Q, c) = \min_{q \in Q}{\var{MinDist}(q, c)}
\end{equation}

In Line~\ref{filterp:add} of Algorithm~\ref{alg:filteringpoint}, a 
point that cannot be pruned is a filter point and is added
into $\sfill$.
First the route ID of the point is found, and inserted into
$\sfill.R$.
Then the point is inserted into the corresponding sorted point list
$\sfill.P$, and each point is affiliated with a list of route IDs
containing it.

Algorithm~\ref{alg:pt} shows how the filtering works.
The filtering of a node is conducted in two steps.
In step 1 (Line~\ref{filter:step1}-\ref{filter:step1-end}), the
filter points $\sfill.P$ are processed to do the filtering.
All points in $\sfill.P$ are sorted by the size of their crossover
route set, and each point is accessed in a descending order.
If a filtering point if found that can filter the node, then all
affiliated route IDs are added to $S$.
If $S$ contains more than $k$ unique route IDs, termination occurs
and the node can be filtered out.
After checking all the filtering points, step 2
(Line~\ref{filter:step2}-\ref{filter:step2-end}) is initiated, and
the routes inside $\sfill.R$ are used for filtering.
Finally, the filtering method based on Voronoi diagrams is employed
(Section~\ref{sec:voronoi}).

\begin{algorithm}[h]
	\KwOut{whether the \var{node} can be filtered}
	\caption{$\mathsf{IsFiltered}(Q$, $\sfill$, $\var{node}$, $k$)}
	\label{alg:pt}
	$\var{S} \leftarrow \emptyset$\;
	\ForEach{$p \in \sfill.P$\label{single:start}\label{filter:step1}}
	{	\tcp{access list points in descending order}
		\If{$\avar{S}.\avar{size}>k$}
		{
			\Return \var{true};
		}
		$\var{label} \leftarrow \var{true}$\;		
		\ForEach{$q \in Q$}
		{
			\If{\avar{node} located in $\hs{p}{q}$}
			{
				$\var{label} \leftarrow false$\;
			}
		}
		\If{$\avar{label}=\avar{true}$}
		{
			$\var{S} \leftarrow \var{S} \cup \crs{p}$\tcp*{crossover route set}
		}\label{single:end}
	}\label{filter:step1-end}
	$S^{'} \leftarrow \sfill.R-S$\;
	\label{filter:step2}
	\ForEach{$\avar{route} \in S^{'}$}
	{
		\If{$S.\avar{size}>k$}
		{
			\Return $\var{true}$;
		}
		
		\tcp{\small see Section~\ref{sec:voronoi}}
		\If{$\mathsf{VoronoiFiltering}(Q, \avar{route} , \avar{node})$}
		{
			$\var{S} \leftarrow \var{S} \cup \{\var{route}\}$\;
		}
		\label{filter:step2-end}
	}
	\Return $\var{false}$;
\end{algorithm}

\subsubsection{Transition Pruning}
\label{sec:tp}
Based on the filter set $\sfill$, entries $e$ from {\sr} are added to a
heap which is sorted by the distance to the query in ascending
order, and checked to see if they can be pruned by $\sfill$ using 
Algorithm~\ref{alg:pt}. 
\scomment{\shane{I don't understand the logic of this sentence at all.
		You take items from $\sfill$, heapify them, and then uses the remaining
		items in $\sfill$ to prune?}}
Algorithm~\ref{alg:pt} uses the candidates in $\sfill$ to check
whether $e$ is located in a filtering space of $Q$.
The transition points that cannot be pruned are inserted into the
candidate set for further refinement.

Algorithm~\ref{alg:rfiltering} describes the procedure to prune the
transition points using the generated filter set $\sfill$ from $\sr$.
It is similar to the filtering method for generating the filtering
set.
The main difference with the traversal of Route R-tree is that only
the unpruned points need to be stored, and the filtering set $\sfill$
is fixed.
As a result, a set of transition points $\scnd$ is obtained which
takes the query routes as $k$ nearest neighbors.

\begin{algorithm}[h]
	\KwOut{$\scnd$: candidate set}
	\caption{$\mathsf{PruneTransition}(\var{root}_t$, $Q$, $\sfill$, $k$)}
	\label{alg:rfiltering}
	$\var{minheap} \leftarrow \varnothing$, $\scnd \leftarrow \varnothing$\;
	$\var{minheap.push}(\var{root}_t)$\;
	\While{$\avar{!minheap.isEmpty()}$}
	{
		$\var{e} \leftarrow \var{minheap.pop()}$\;
		\If{$\avar{e}$ is a Node}
		{
			\If{$\mathsf{isFiltered}$($Q$, $\sfill$, $\avar{e}$, $k$)}
			{
				\textbf{continue}\;
			}
		}
		\eIf{$\avar{e}$ is a point}
		{
			$\scnd.\var{add}(e)$\;
		}
		{ 
			\ForEach{child $c$ of $\avar{e}$}
			{
				$\var{minheap.push}(c,\var{MinDist}(c, Q))$\;
			}
		}
	}
	\Return $\scnd$\;
\end{algorithm}

\subsubsection{Verification}
\label{sec:verify}
The verification mainly uses the filtered node set $\srefine$ during
the traversal of $\fr$ to find $\sfill$ in
Algorithm~\ref{alg:filteringpoint}.
It can be divided into two steps.
First, the nodes in $\fr$ encountered during the filtering phase are
kept in $\srefine$ in Line~\ref{alg:fil:refine} of
Algorithm~\ref{alg:filteringpoint}.
The verification algorithm runs in rounds.
In each round, one of the nodes in $\srefine$ is opened and its
children are inserted into $\srefine$.
During each round, the nodes and points in $\srefine$ are used to
identify the candidates that can be verified using $\srefine$, which
are the nodes confirmed as {\rknnt} or guaranteed not to be {\rknnt}.
Such candidates are verified and removed from $\scnd$.
The algorithm terminates when $\scnd$ is empty.
The result set is then stored for a second round of verification.

To verify a candidate effectively, if more than $k$ routes are found
in $\srefine$ which are closer to the query than the candidates, then
it can safely be removed from $\scnd$.
Hence, we maintain a set to store the unique IDs of these routes when
every candidate point is checked in $\scnd$.
The route IDs are found, and the set is updated using $\nlist$ when
new filtering points or nodes from $\srefine$ are found.
When the number of IDs in a set is greater than $k$, it can be
removed from $\scnd$.

After finding the transition points for the routes, as they will take
the query as $k$ nearest neighbors, then for the $\exists$\rknnt, the
transition ID for all remaining points can be returned as the final
result $\sresult$ in the second step.
While $\forall$\rknnt, if a transition only has one point in the
result set, then it will be pruned, and only the transitions which
have both points in the result will be considered as the real result
and added to $\sresult$.

\section{Optimizations on Filtering Process}\label{sec:optimization}
\subsection{Voronoi-based Filtering}
\label{sec:voronoi}
One problem of the filtering method in Algorithm~\ref{alg:pt} is that
the filtering space obtained from a single point and the query is
usually very small.
For example in Figure~\ref{fig:onepoint}, $\var{MBR}_1$ cannot 
be pruned, so it needs to load
$\var{MBR}_2$ and $\var{MBR}_3$ to perform further checks, which
require additional pruning time.
To further enlarge the pruning space, the available filtering points
in a single route can be used rather than a single point to perform
the pruning, namely, $\sfill$ can be used for additional pruning.
Given a query and a filtering route $\route$, a larger filter space
can be explored.
To find this area, Voronoi cells can be used.

To accomplish this, a plane is partitioned with points into several
convex polygons, such that each polygon contains exactly one
generating point, and every point in a given polygon is closer to its
generating point than to any other.
The convex polygon of one point is called as the \emph{Voronoi cell},
and the point is called the \emph{kernel} of this cell.
By plotting the Voronoi diagram $\mathcal{V}_{\route, Q}$ between the
query $Q$ and a filtering route $\route$, as shown in 
Figure~\ref{fig:half}. The Voronoi cell
$\mathcal{V}_{\route, Q}[p]$ of the route $\route$ can be found, and
any point inside these cells is closer to the filtering route than
the query.
Furthermore, if a node does not intersect with any cell of the query,
then any point inside this node will be closer to the filtering route
than the query.
If a node can find more than $k$ such filtering routes, then the node
can be pruned.
\begin{figure}
	\centering
	\vspace{-1em}
	\includegraphics[height=4cm]{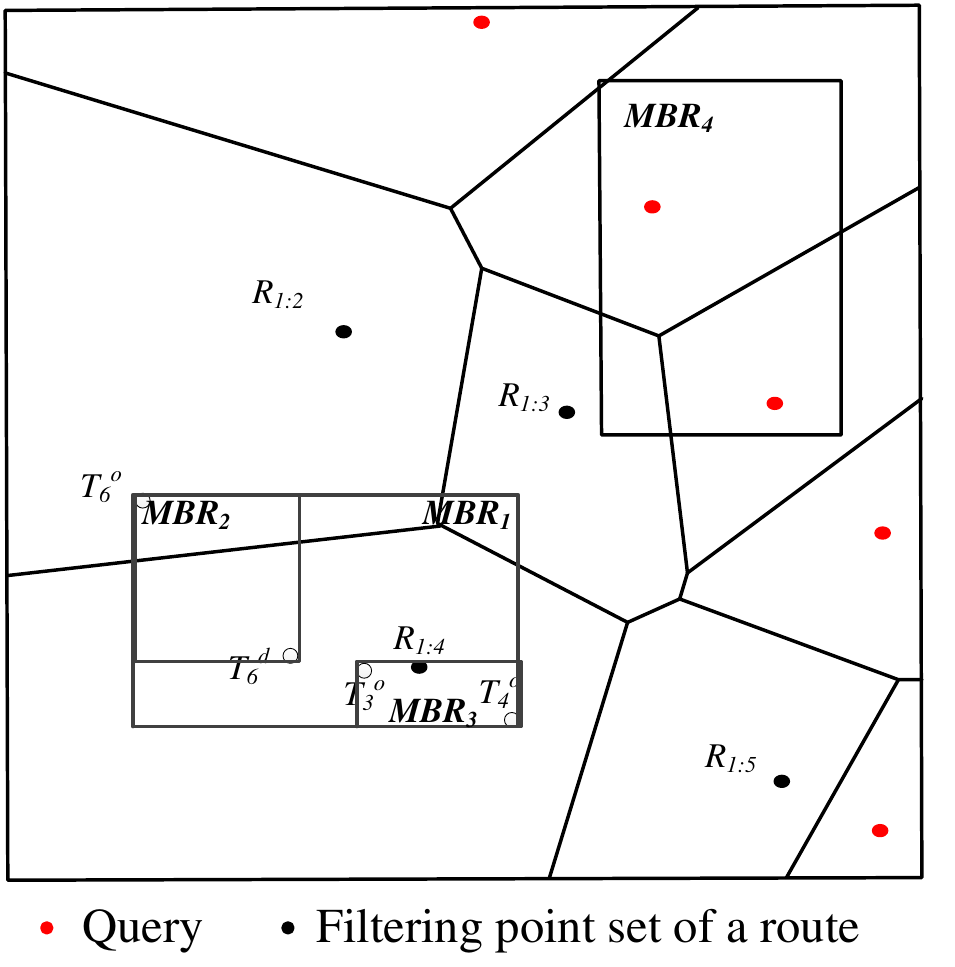}
	\vspace{-1em}
	\caption{Pruning based on the Voronoi diagram of a query (red) and 
		a filtering route composed of $4$ points (black).}
	\label{fig:half}
\end{figure}

\begin{definition}{(\textbf{Voronoi Filtering Space})}
	\label{de:mfs}
	Given a filtering route $\route$ and query $Q$,
	we define the Voronoi filtering space as:
	\begin{equation}
	\hs{\route}{Q} = \bigcup\limits_{p \in \route} {\mathcal{V}_{\route, Q}[p]}
	\end{equation} 
	which is a union of the Voronoi cells of all points from
	$\route$, and $\mathcal{V}_{\route, Q}$ is the Voronoi
	diagram of union of points from $\route$ and $Q$.
\end{definition}


Any transition point inside $\hs{\route}{Q}$ cannot have $Q$ as the
nearest neighbor.
As shown in Figure~\ref{fig:half}, for any point in the
Voronoi filtering space, it can find a point in the filtering
route which is closer than any point in the query.
Hence, two points in a transition can both find a point in
the filtering route rather than the query, so the transition
point will not the take query as the nearest neighbor.

The filtering route $\route$ is used to further
prune the transition point if it cannot be pruned by the filter points
in $\route$ one by one as shown in Line~\ref{single:start}
-\ref{single:end} of Algorithm~\ref{alg:pt}.
After scanning all the filtering points of a route in $\sfill.P$, we
will use the Voronoi filtering space of the route for the query to
prune the transition points, where the space has been created after
getting the filtering route set, then the pruning space will be
larger, and we can find one more route which is closer to 
the entry than the query if it can prune the
entry.

For example, consider the $4$ points belonging to a same route $R_1$
to prune the transition points in Figure~\ref{fig:half}.
The filtering space is larger than the area shown in
Figure~\ref{fig:onepoint}, and $\var{MBR}_1$ is entirely located
within the filtering space, so it can be pruned from consideration.
Since the Voronoi diagram can be produced at the same time as when
the perpendicular bisectors from query to every filtering point are
computed, then there is no additional cost to generate the Voronoi
information.
This additional pruning rule improves the probability of a node being
pruned.

\subsection{Divide \& Conquer Method}
\label{sec:dc}
Note that the processing of the proposed method will be complex when
the query has many points.
The main reason is that a node has to be filtered by every query
point, and the probability of a point being pruned will be lower
when the query length is larger.
To alleviate this problem, we introduce a divide-and-conquer method based
on our processing framework.

\begin{lemma}
	\label{le:trw}
	The {\rknnt} of a multi-point query is the union of the {\rknnt}
	of all points in a query:
	\begin{equation}
	\rknnt(Q)=\bigcup_{q_i\in Q}{\rknnt{(q_i)}}
	\end{equation}
\end{lemma}
\begin{proof}
	For a transition, if it takes a query point as a $k$ nearest
	neighbor, then it must be a \rknnt~result for $Q$, 
	so $\bigcup_{q_i\in Q}{\rknnt{(q_i)}} \subseteq \rknnt(Q) $.
	For each transition in $\rknnt(Q)$, it must take one query
	point in $Q$ as the {\knn}, 
	then $\rknnt(Q) \subseteq \bigcup_{q_i\in Q}{\rknnt{(q_i)}}  $.
	Based on above two observations, 
	$\rknnt(Q)=\bigcup_{q_i\in Q}{\rknnt{(q_i)}}$.
\end{proof}
\scomment{\shane{I'm not entirely convinced by this proof. It seems incomplete
		to me. Is this an observation or a lemma?}}

Based on this observation, a divide and conquer framework is proposed
that uses multiple {\rknnt} searches which were introduced in
Section~\ref{sec:baseline}.
The main idea is that {\rknnt} search is performed for every query
point to find a candidate transition point set for every query point
first, and then the transitions containing these points are merged to
get the final transition result.

Even though an $\rknnt$ query mainly targets a route query, it can
process single-point queries as well since every step in the
algorithm does not require that the query have more than one point.
According to Definition~\ref{de:fs}, the filtering space will be the
largest when there is only one query point, so the pruning efficiency
will be the highest when compared with any multi-point query which
extends from this single query point.

\section{Optimal Route Planning}
\label{sec:maxrknnnt}
In this section, we present a solution to the route planning 
problem with a distance threshold
based on {\rknnt}.
We first define a new query called {\maxrknnt}.
A baseline method is proposed first, and then an efficient search
method based on pre-computation and pruning is described.
\begin{figure}[t]
	\centering
	{
		\includegraphics[height=3cm]{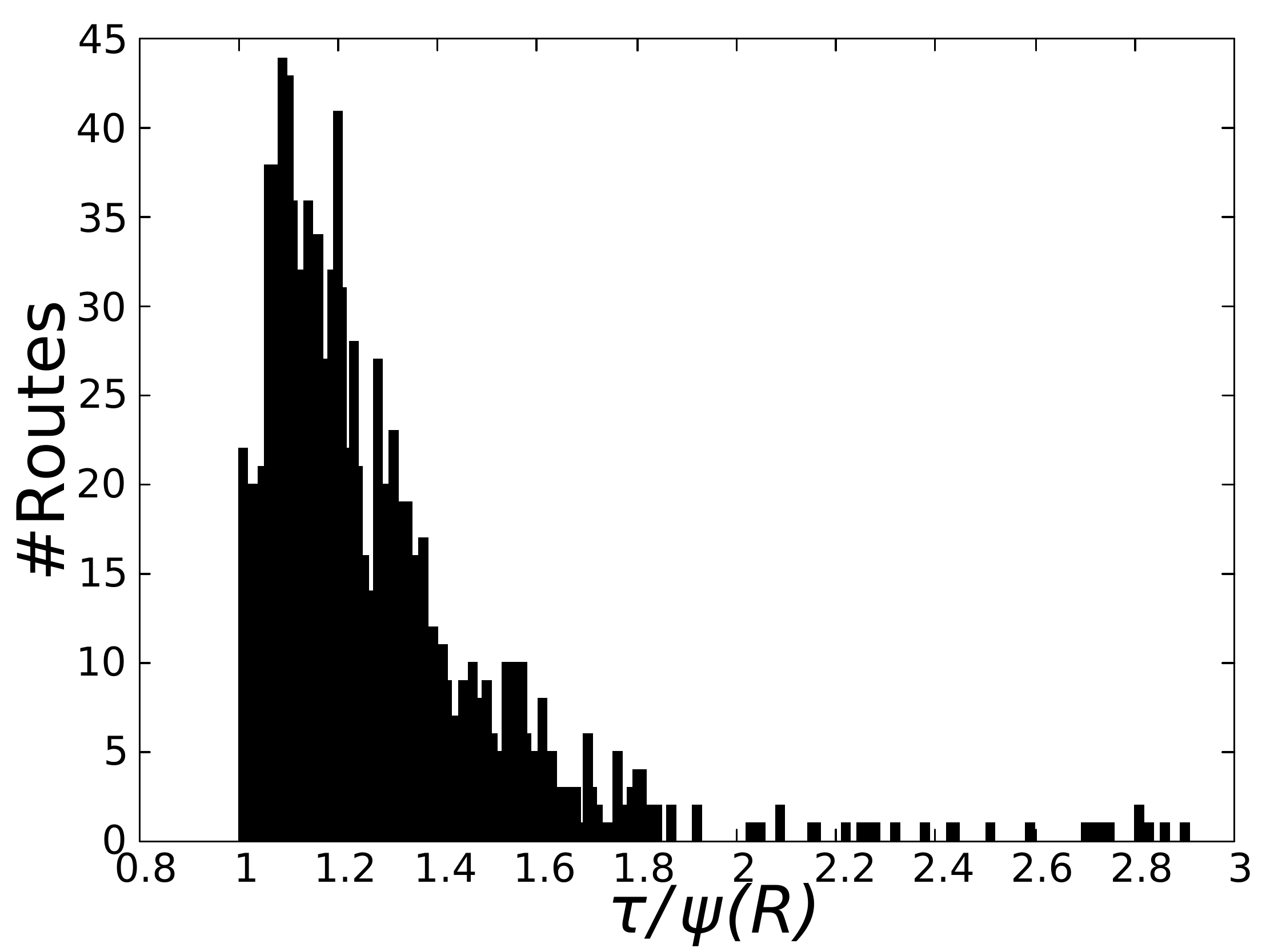}
		\includegraphics[height=3cm]{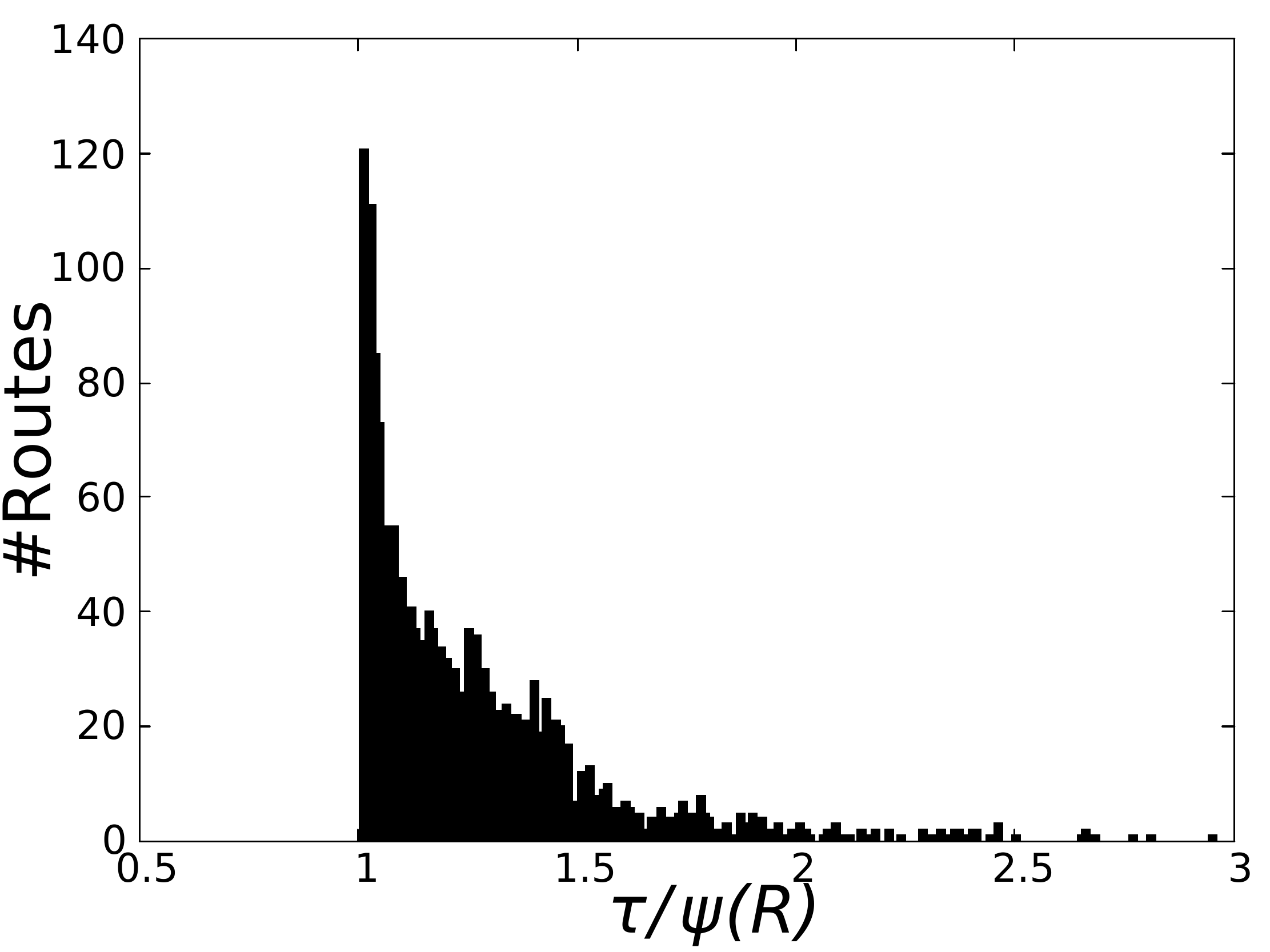}
	}
	\vspace{-1.5em}
	\caption{Frequency histogram of the ratio between travel distance 
		and straight-line distance for all routes in LA and NYC.}
	\label{fig:ratio}
	\vspace{-1em}
\end{figure}

\subsection{Maximizing {\large \rknnt} in a Bus Network}
\label{sec:mrbn}
In bus route planning, the goal is how to attract the maximum number
of passengers within a given distance threshold, since a single bus
cannot cover all stops in a city. For Uber drivers, such a route
also means high possibility to get high profit by attracting bus
passengers to take the Uber to travel alternatively. Next, we
will introduce the maximizing \rknnt~in the bus network.

Here we take the real bus networks in NYC and LA as an example.
Figure~\ref{fig:ratio} shows that the ratio between the travel
distance and the straight line distance between start and end 
bus stops, and does not exceed $2$ in most bus routes.
Hence, such a distance constraint always exists in real-life route
planning.

We first cast the existing bus network as a {\em Weighted Graph}.
\begin{definition}{\textbf{(Weighted Graph)}}
	\label{de:wg}
	$\graph=(E,V)$ is a weighted graph, where $V$ is the vertex set and
	$E$ is a set of edges which connect two vertices among $V$.
	A route in $\graph$ is a sequence of vertices $R=(v_{1},v_{2},\ldots
	,v_{n}) \in V\times V\times \cdots \times V$ such that $ v_{i}$ is
	adjacent to $v_{i+1}$ for $1\leq i<n$, $v_1$ and $v_n$ are the start
	and end vertex respectively.
\end{definition}

Given a route $\route$, $\traveldistance{\route}$
is the travel distance starting from start to end through every
vertex in the route:
\begin{equation}
\label{equ:td}
\traveldistance{\route}= \sum_{p_i\in \route\space \&\space i\in[1, n-1]}{\avar{distance}(p_i,p_{i+1})}
\end{equation}

Recall Definition~\ref{def:rknnt}, given a route $\route$ in
$\graph$, among the transition set $\ts$, the $\rknnt$ of $\route$
can find all transitions that would choose it as a {\knn}.
The passengers who are likely to take $\route$ are
the {\rknnt} set of $\route$.
Let $\omega(\route) = \rknnt(\route)$ for simplicity.
We now formally define the Maximizing \rknnt~($\maxrknnt$) problem
for route planning.

\begin{definition}{\textbf{(\maxrknnt)}}
	\label{de:maxrknnt}
	Given a distance threshold $\tau$, a starting vertex $\orgin$ and a
	destination vertex $\destination$, $\maxrknnt(\orgin, \destination, \tau)$
	returns an optimal route $\planroute$ from $S_{se}$ such that $\forall
	\planroute^{'}\in S_{se}-\planroute$, $|\omega(\planroute)| \ge
	|\omega(\planroute^{'})|$ and $\traveldistance{\planroute} \le \tau$
	, where $S_{se}$ is the set of all possible routes in $\graph$ that
	share same start and end vertex.
\end{definition}

The definition of $\minrknnt$~can be defined by changing
$|\omega(\planroute)| \ge |\omega(\planroute^{'})|$ to
$|\omega(\planroute)| \le |\omega(\planroute^{'})|$ in
Definition.~\ref{de:maxrknnt}.
In this paper, we propose a search algorithm which can solve both
{\maxrknnt} and {\minrknnt}.
By default, we choose {\maxrknnt} for ease of illustration.

\myparagraph{Baseline}
A brute force method for {\maxrknnt} is to find all the candidate
routes which meet the travel distance threshold constraint.
This can be done by extending the {\em k shortest path} method proposed
by {\cite{Yen1971,Martins2003}} with a loop to find the sub-optimal
route until the distance threshold $\tau$ is met.
Then a {\rknnt} query is ran for each candidate and the one with
maximum number of results as the optimal route is selected.
We call this method {\textbf{BF}}.
Recall the query in Figure~\ref{fig:graph}, where almost all routes
such as $\overline{abej}$,$\overline{acej}$ and $\overline{acehj}$
will be candidates.

However, the performance of {\rknnt} decreases as the number of
points increases, which is discussed in more detail in
Section~\ref{sec:em}.
For the bus route planning, it may be tolerable to wait for a few
seconds to conduct {\maxrknnt} query.
However, for real time queries, like identifying profitable routes
for Uber drivers, this method will not work well.
To achieve better performance, an efficient route searching algorithm
is proposed based on the pre-computation of the {\rknnt} set for each
vertex in $\graph$.
\scomment{
	\shane{Technically, if you're pre-computing things, how can it be
		realtime?}
	\sheng{Since the graph is static, we can pre-compute the distance
		and rknnt in advance, when query is coming, we can achieve more efficient
		search based those pre-computed thing.}}

\subsection{Our Solution}
According to Lemma~{\ref{le:trw}}, the query $Q$ can be decomposed
into a set of $|Q|$ queries, which means that we can get the
pre-computed {\rknnt} set for every vertex, and conduct a union
operation on all vertices in a route to get the final {\rknnt} set
for that route.

By using the above property, we introduce a pre-computation based
method with a fixed $k$ which provides better performance.
Note that even though $k$ should be fixed in the pre-computation,
multiple datasets of representative $k$ can be generated in advance
to meet different requirements.

\subsubsection{Pre-computation}
For every vertex in $\graph$, an $\rknnt$ query is ran, and the
result stored.
A pre-computed matrix $\mathcal{M}_\psi[i][j]$ is created which
stores the pre computed all-pair shortest distance for all vertexes in
$\graph$ using the Floyd-Warshall algorithm~{\cite{Burfield2013}}.
The details of pre-computation can be found in
Algorithm~{\ref{alg:pre}}.
\scomment{
	\shane{Umm, isn't Floyd-Warshall a cubic algorithm?
		This does not seem like a scalable algorithm at all.}
	\sheng{It is pre-computed, which is independent with query.}}
\begin{algorithm}[h]
	\KwOut{$\graph.V$: the vertexes with \rknnt~set}
	\caption{$\mathsf{Precomputation}(\graph$, $\rs$, $\ts$, $k$)}
	\label{alg:pre}
	$\var{root}_{\svar{t}} \leftarrow \var{CreateIndex}(\ts)$\tcp*{ \small root of \fr}
	$\var{root}_{\svar{r}} \leftarrow \var{CreateIndex}(\rs)$\tcp*{ \small root of \sr}
	\ForEach{vertex $v \in \graph$}
	{
		$\sresult \leftarrow \rknnt(v, \var{root}_{\svar{r}},
		\var{root}_{\svar{t}})$\tcp*{\small call Algorithm~\ref{alg:fw} by query $v$}
		\scomment{
			\ForEach{Transition $t \in \sresult$ }
			{
				$\var{root}_{\svar{t}} \leftarrow label(\var{root}_{\svar{t}}, t.p_o)$\;
				$\var{root}_{\svar{t}} \leftarrow label(\var{root}_{\svar{t}}, t.p_d)$\tcp*{\small labeled as visited}
			}}
			$\graph.V.\rknnt(v)\leftarrow \sresult$\tcp*{\small update the set on vertex}
			\ForEach{vertex $v^{'} \in \graph-v$}
			{
				$\mathcal{M}_\psi[v][v^{'}] \leftarrow \mathsf{ShortestDistance}(\graph, v, v^{'})$\;
			}
		}
		\Return $\graph.V$ \;
	\end{algorithm}
	
	With the pre-computed {\rknnt} set, we can further improve the
	performance of the baseline method {\textbf{BF}}.
	After getting all candidate routes that do not exceed the distance
	threshold, the {\rknnt} set of each route can be found by performing
	a union operation on the sets.
	Compared with the baseline method, the on-the-fly {\rknnt} query is
	replaced with pre-computation, and the running time for the search is
	reduced to the search time of $k$ shortest path search.
	However, it is still possible to leverage distance constraints and
	dominance relationships to prune additional routes in advance.
	
	\begin{figure}
		\vspace{-1em}
		\begin{minipage}[b]{.5\linewidth}
			\centering
			\label{fig:expansion}
			\includegraphics[width=5.3cm]{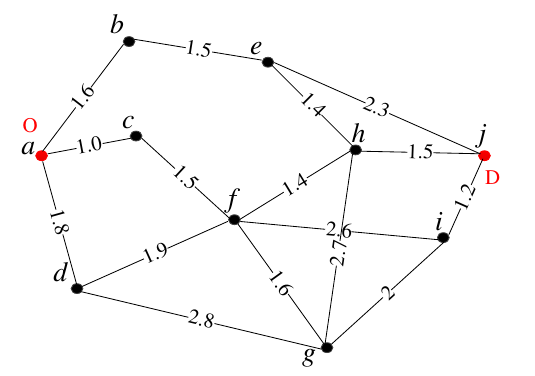}
			\vspace{-3mm}
		\end{minipage}~~~~~~~~~~~~~
		\begin{minipage}[b]{.45\linewidth}
			\label{fig:rknntweight}
			\begin{tabular}{|c|c|}
				\hline
				Vertex & \rknnt~Set \\
				\hline
				$a$      & $T_{1}^o$         \\
				\hline
				$b$      &    $T_{1}^d$          \\
				\hline
				$c$      &       $T_{1}^d, T_{3}^o, T_{4}^o$     \\
				\hline
				$d$      &      $T_5^o$        \\
				\hline
				$e$      &     $T_{2}^o$         \\
				\hline
				$f$      &      $T_2^o, T_3^d,T_4^d$        \\
				\hline
				$g$      &     $T_5^o$         \\
				\hline
				$h$      &       $T_{2}^d$        \\
				\hline
				$i$    &        $T_6^o$      \\
				\hline
				$j$     &       $T_6^d$      \\
				\hline
			\end{tabular}
			\vspace{-1em}
		\end{minipage}
		\caption{An exemplar graph with a query $(a,j,6)$ where $a$ and $j$ are
			the start and end vertexes, $\tau=6$ is the distance threshold, and
			the table shows the {\rknnt} set for each vertex.}
		\label{fig:graph}
		\vspace{-2em}
	\end{figure}
	
	\begin{example}
		As shown in Figure~\ref{fig:graph}, the red points are the start
		$O=a$ and end $D=j$ respectively.
		A query formed by these two points and $\tau=6$ return the route with
		largest {\rknnt} set, where the number on each edge is the distance
		between two vertices, and the label is the vertex ID.
		The table shows the pre-computed $\rknnt$ set for each vertex.
		So, $\omega(\overline{acfhj})=\{T_1,T_2,T_3,T_4,T_6\}$ and
		$\traveldistance{\overline{acfhj}}=1+1.5+1.4+1.5=5.4$.
	\end{example}
	
	\subsubsection{Route Searching by Pruning}
	After getting the {\rknnt} set for every vertex in the graph $\graph$,
	Algorithm~\ref{alg:csp} can be ran to get the optimal route based on the
	pre-computed Euclidean distance of every edge.
	Specifically, the neighbor vertices are accessed around the starting
	point, and two levels of checking are performed to see whether the
	current partial route $\currentroute$ is feasible.
	If it is, it is inserted into the priority heap $\mathcal{Q}$, and
	the partial route is increased until it meets the end point
	$\destination$ and has the maximum result set size.
	Specifically, the two checking functions work as below: 
	
	\begin{algorithm}
		\KwOut{$\planroute$: the optimal route}
		\caption{{\maxrknnt}($o$, $d$, $\tau$)}
		\label{alg:csp}
		\If{$\mathsf{checkReachability}(\orgin, \destination,\tau)$}
		{
			\Return $\varnothing$\;	
		}
		$\planroute \leftarrow \varnothing$,
		$\currentroute \leftarrow \{\orgin\}$\;
		$\traveldistance{\currentroute} \leftarrow 0$\tcp*{\small travel distance}
		$\omega(\currentroute) \leftarrow \graph.V.\rknnt(\orgin)$\tcp*{\small \rknnt~set}
		$\queue \leftarrow \varnothing$ \tcp*{ \small queue stores the partial routes}
		$push(\queue, \{\currentroute, \traveldistance{\currentroute}, \omega(\currentroute)\})$\;
		${\var{max}} \leftarrow |\omega(\currentroute)|$ \label{alg:max:max}\;
		\While{$\queue\ne \varnothing$}
		{
			$ \{\currentroute , \traveldistance{\currentroute},  \omega(\currentroute)\} \leftarrow pop(\queue)$\;
			$v_i \leftarrow \var{GetEnd}(\currentroute)$\;
			\ForEach{$v_j\in \avar{Neighbor}(\graph, v_i)$}
			{
				\If{$\mathsf{checkReachability}(v_j,d,\tau-\traveldistance{\currentroute})$}
				{
					\If{$\mathsf{checkDominance}(o, v_j, \traveldistance{\currentroute} ,  \omega(\currentroute)  )$\label{alg:max:domi}}
					{
						{
							$S \leftarrow \var{Update}(\mathcal{DT}[d], \traveldistance{\currentroute}, \omega(\currentroute)))$\;
							\ForEach{$\avar{candidate}\in S$}
							{
								$\var{Delete}(\mathcal{Q},\var{candidate})$\;
							}
							${\currentroute} \leftarrow \currentroute \cup \{v_j\}$\;
							$\traveldistance{\currentroute} \leftarrow \traveldistance{\currentroute}+\traveldistance{v_i,v_j}$\;
							$\omega(\currentroute) \leftarrow \omega(\currentroute) \bigcup \graph.V.\rknnt(v_j)$\;
							$push(\queue, \{\currentroute, \traveldistance{\currentroute},  \omega(\currentroute)\})$\;		
						}
					}
				}
			}
			\If{$\avar{GetEnd}(\currentroute) = \destination$}
			{
				\If{$|\omega(\currentroute)|> {\var{max}}$ \label{alg:max:compare}}
				{
					\tcp{\small new optimal route}
					$\planroute \leftarrow \currentroute$\;
					${\var{max}} \leftarrow |\omega(\currentroute)|$\;
				}
			}
		}
		\Return $\planroute$\;
	\end{algorithm} 
	
	\myparagraph{$\mathsf{checkReachability}$} 
	This pruning function checks whether the current route meets the
	distance constraint -- namely that the distance from the current
	vertex to the end vertex is less than
	$\tau-\traveldistance{\currentroute}$.
	When $\md{v_j}{d}>\tau-\traveldistance{\currentroute}$, it will
	return false and move to next neighbor of vertex $v_i$ in $\graph$.
	
	\myparagraph{$\mathsf{checkDominance}$}
	This pruning function exploits the dominance relationship between two
	partial routes.
	If a partial route exists that ends at the same vertex and has a
	short route and a larger {\rknnt} set, then it can dominate the
	current route.
	Specifically, a dominating lemma is introduced which works for both
	$\forall${\rknnt} and $\exists$\rknnt.
	\begin{lemma}
		Given two partial routes $R_1^*$ and $R_2^*$ which have the same start
		and end, $R_1^*$ dominates $R_2^*$ in {\maxrknnt} ($R_2^*$ dominates
		$R_1^*$ in {\minrknnt}) when
		$|\traveldistance{R_1^*}|<|\traveldistance{R_2^*}|$ and
		$|\forall\rknnt(R_1^*)|>|\exists\rknnt(R_2^*)|$.
	\end{lemma}
	\begin{proof}
		For purposes of proving the lemma, we use $\rknntweight{\route}$ and
		$\omega^*(\route)$ to represent $\exists\rknnt(\route)$ and
		$\forall\rknnt(\route)$ to distinguish them.
		Given any partial route $R^{'}$ which starts at $v_j$ and ends at
		$d$, $R_1^*$ and $R_2^*$ can be connected to form two complete routes
		$R_1$ and $R_2$.
		1) For $\exists$\rknnt, if $|\omega^*(R_1^*)|>|\omega(R_2^*)|$, then
		$|\omega(R_1)|\ge |\omega^*(R_1^*)| +|\omega(R^{'})|$, as there is no
		intersection between $\omega^*(R_1^*)$ and $\omega(R^{'})$ because
		$\tran \in \omega^*(R_1^*)$ is the set of transitions that have
		{\knn} in $R_1^*$ for both origin and destination points.
		Given that $|\omega(R_2)| \le |\omega(R_2^*)| + |\omega(R^{'})|$,
		$|\omega(R_1)|>|\omega(R_2)|$, while
		$\traveldistance{R_1^*}<\traveldistance{R_2^*}$, 2)
		$\forall$\rknnt, $|\omega^*(R_2)| \le
		|\omega(R_2^*)|+|\omega^*(R^{'})|$, while $|\omega^*(R_1)| \ge
		|\omega^*(R_2^*)|+|\omega^*(R^{'})|$, so $|\omega^*(R_1)|
		>|\omega^*(R_2)|$.
		Without further spreading, we can see the priority relationship
		between $|\omega(R_1^*)|$ and $|\omega(R_2^*)|$ holds.
	\end{proof}
	
	In Algorithm~\ref{alg:csp}, a dynamic table $\mathcal{DT}$ is
	maintained to store the pairs for every vertex accessed, and updates
	continue when new feasible partial routes are explored during the
	search.
	This is used to compare the {\rknnt} set and the travel distance of
	partial routes.
	The entry for a vertex $v$ inserts a partial route $R^*$ which ends
	at $v$ when an existing partial route cannot be found which dominates
	$R^*$.
	After insertion, old entries in $\mathcal{DT}$ that are dominated by
	the new route $R^*$ are removed.
	If a new one is found that dominates $R^*$, the loop terminates, and
	the next partial route is processed.
	
	\begin{example}
		In Figure~\ref{fig:graph}, $\{\{a\},0,20\}$ is added to the queue
		$\mathcal{Q}$ after checking the reachability from $a$ to $j$ by
		comparing the pre-computed shortest distance with $\tau$.
		Then, pop the queue $\mathcal{Q}$ to get the partial route $R$.
		Next, the last point $a$ of $R$ is checked to see if its neighbor $b$ can be
		reached, and it can since
		$\traveldistance{\overline{bej}}=3.8<(6-1.6)$.
		So $\{\{a,b\},1.6,\{T_1\}\}$ is added to $\mathcal{Q}$.
		Similarly, $\{a,c\}$ is inserted into $\mathcal{G}$.
		$\{a,d\}$ cannot be enqueued as the shortest distance from $d$ to $j$
		is $\traveldistance{\overline{dfhj}}=5.2>(6-1)$.
		$\{\{a,b,e\},3.1,\{T_1,T_2\}\}$ and
		$\{\{a,c,e\},2.6,\{T_1,T_2,T_3,T_4\}\}$ are enqueued and
		$\mathcal{DT}[e] ={\{\{a,b,e\},3.1,\{T_1,T_2\}\}}$ is updated.
		Further, $\{\{a,c,f,h\},3.9,\{T_1,T_2,T_3,T_4\}\}$ is enqueued.
		$\{\{a,b,e,h\},4.5,\{T_1,T_2\}\}$ has a greater travel distance, and
		$\omega(\overline{abeh})=\{T_1,T_2\}$, and
		$\omega^*(\overline{acfh})=\{T_1,T_2,T_3,T_4\}$, so
		$|\omega^*(\overline{acfh})|>|\omega(\overline{abeh})|$,
		$\overline{acfh}$ dominates $\overline{acfh}$.
		Based on this extension in the graph, when $\mathcal{Q}$ is empty,
		the algorithm terminates.
	\end{example}
	
	For {\minrknnt}, Line~\ref{alg:max:max} is changed to
	$\var{max}\leftarrow \infty$, and Line~\ref{alg:max:compare} is
	changed to $|\omega(\currentroute)|< {\var{max}}$.
	Moreover, one additional check called $\mathsf{checkBounds}(\var{max},
	\omega(\currentroute))$ after Line~\ref{alg:max:domi} in
	Algorithm~\ref{alg:csp} must be added.
	Given a partial route $R^*$ and the existing optimal route $R$
	and $\var{max}$, $R^*$ can be discarded when
	$|\omega(R^*)|>\var{max}$ as $R^*$ can not beat
	the existing optimal route $R$.

\section{Experiments}
\label{sec:experiment}
\subsection{Experimental Setup}
\begin{figure}[!htbp]
	\centering
	{
		\includegraphics[height=2cm]{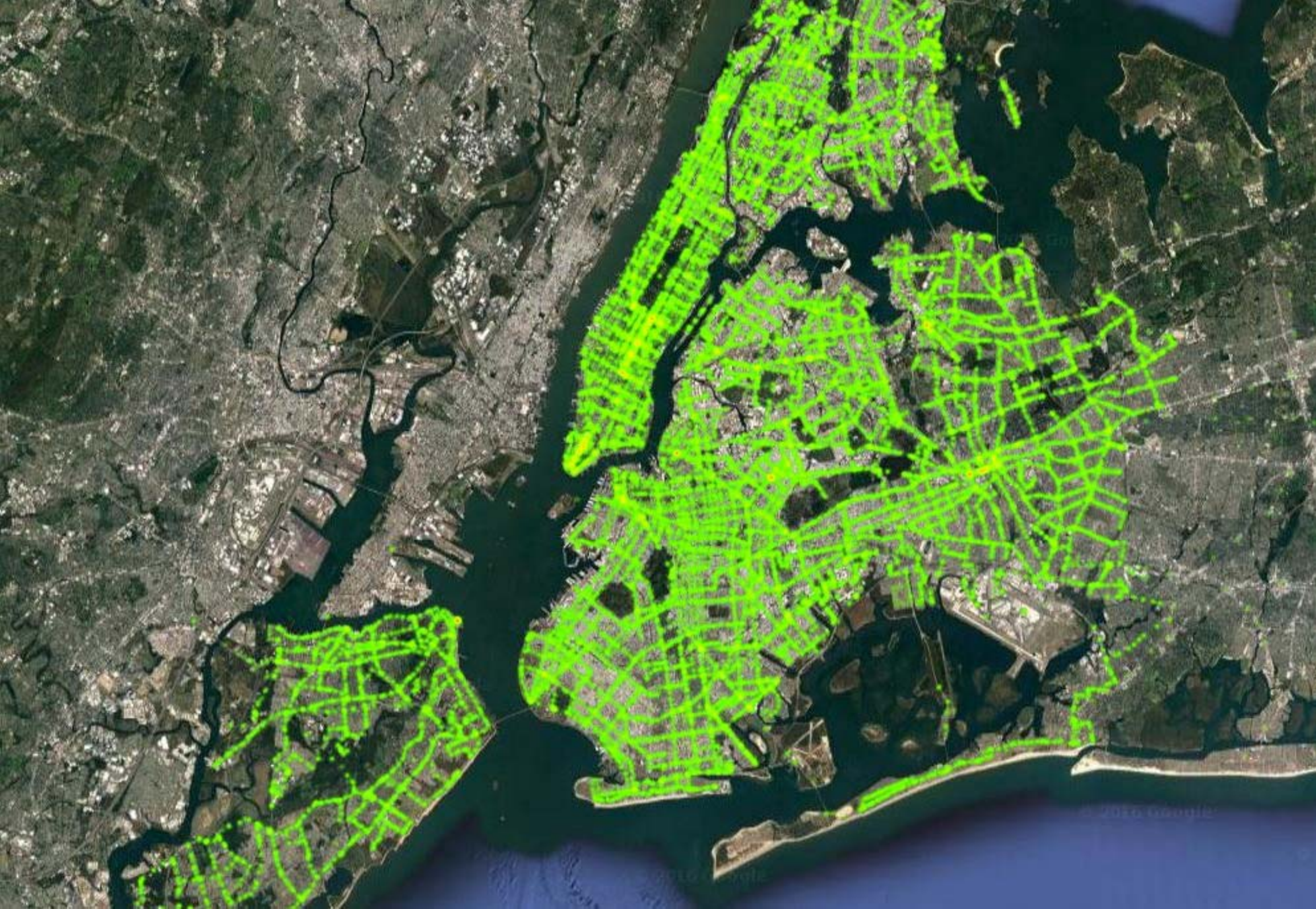}
		\includegraphics[height=2cm]{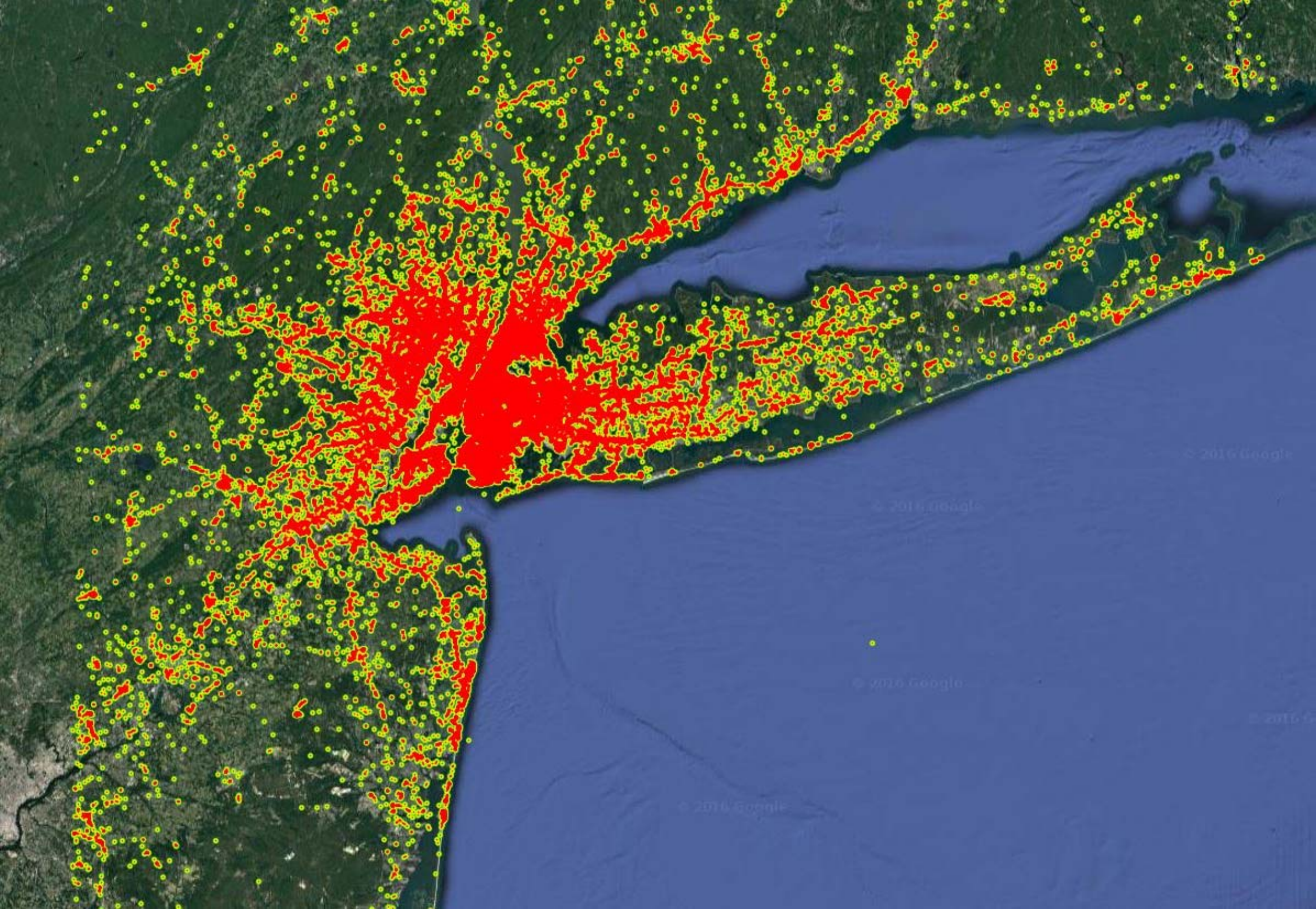}
	}\\
	~\centering
	{
		\includegraphics[height=2cm]{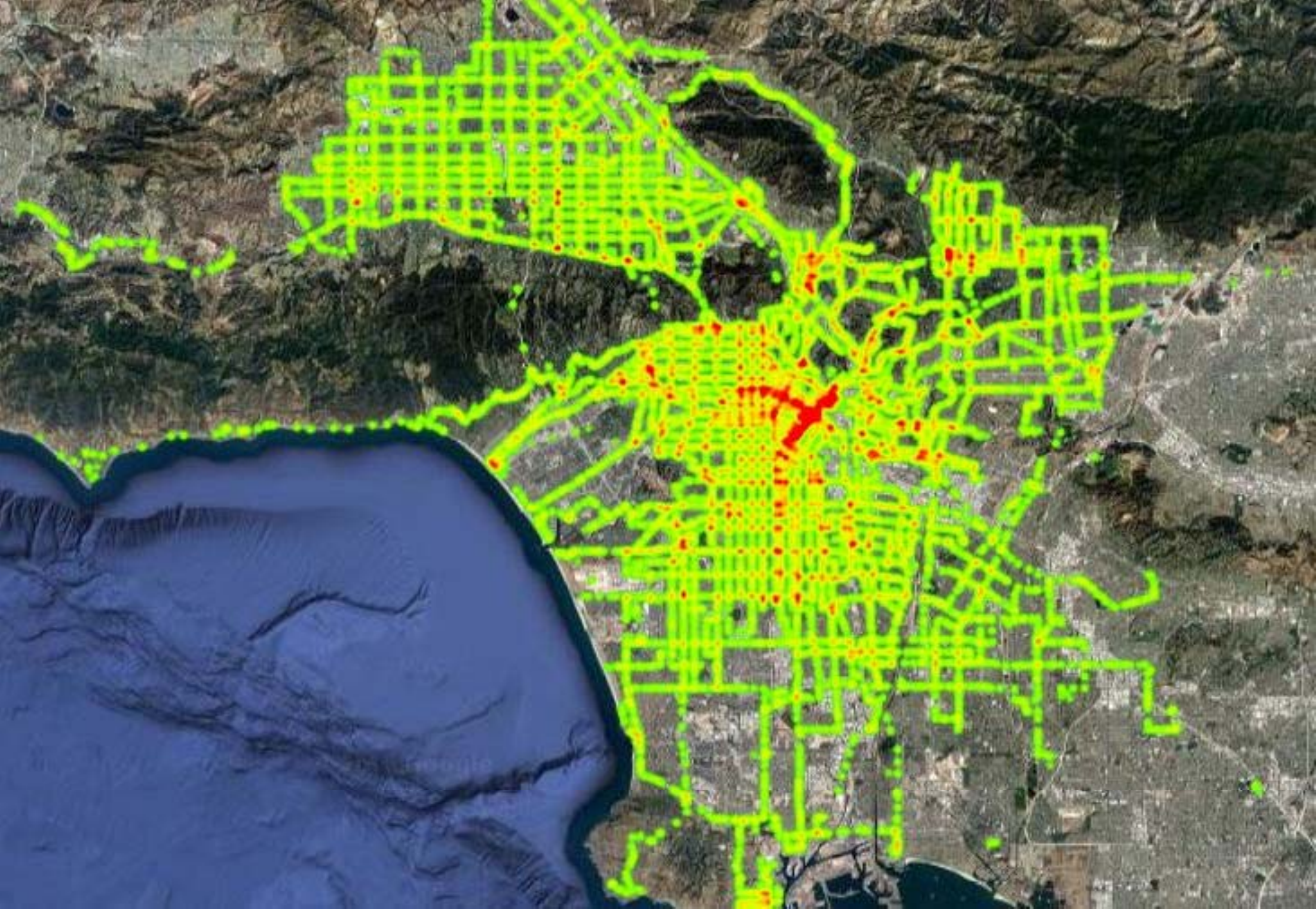}
		\includegraphics[height=2cm]{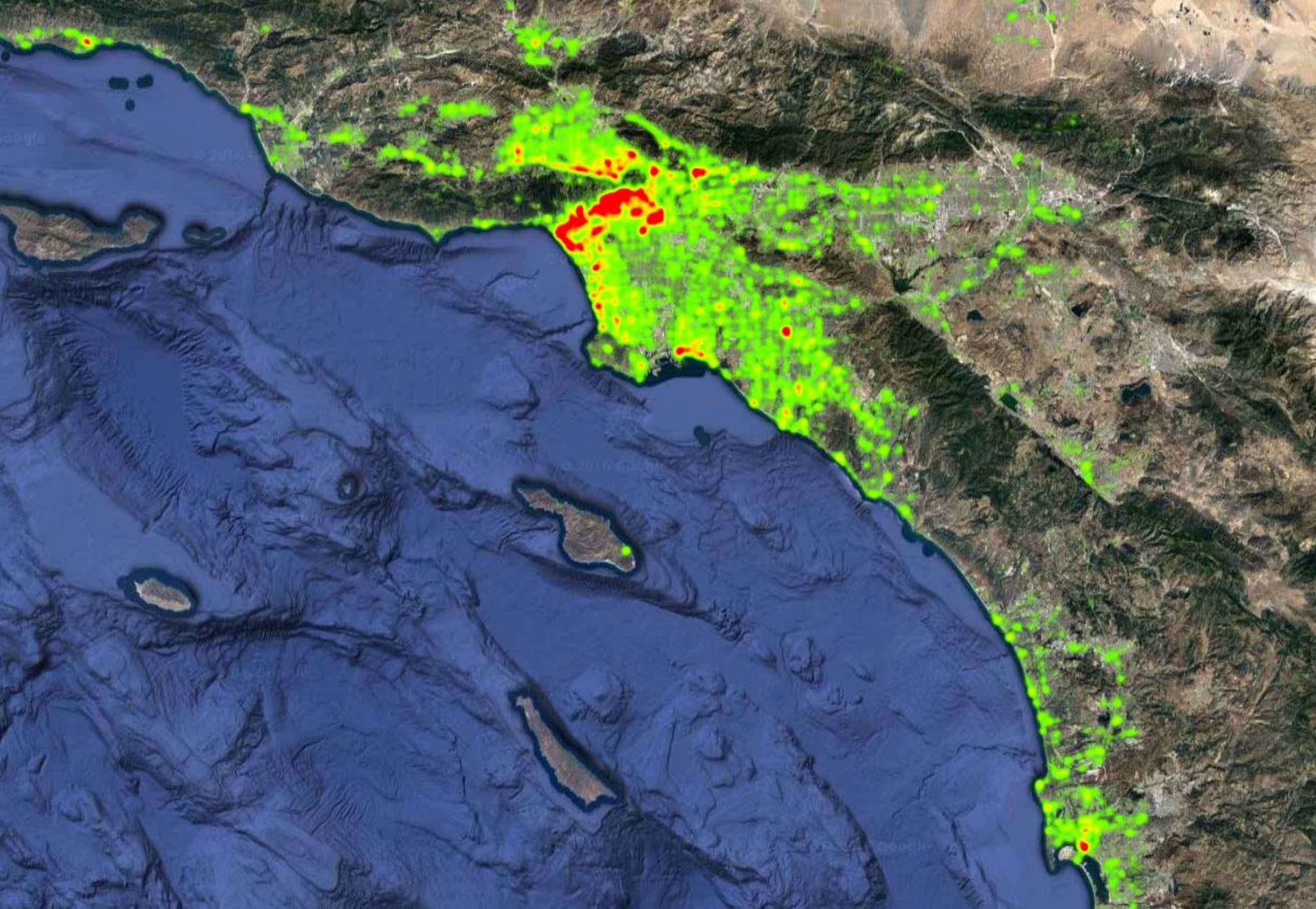}
	}
	\vspace{-1em}
	\caption{The heatmap of the bus route dataset (left) and the transition dataset (right) in NYC (up) and LA (down).}
	\label{fig:heatmap}
\end{figure}

\begin{table}[!htbp]
	\centering
	\caption{Route Datasets.}
	\label{tab:route}
	\begin{tabular}{|c|c|c|c|c|c|}
		\hline
		\textbf{Dataset} &$\left|\rs\right|$ & $|\mathcal{G}.E|$ & $|\mathcal{G}.V|$  \\ \hline \hline
		{\em LA-Route}  &  $1,208$ &  $72,346$   &  $14,119$  \\ \hline
		{\em NYC-Route}   & $2,022$ &  $61,118$  &  $16,999$ \\ \hline
	\end{tabular}
	\vspace{-1em}
\end{table}

\begin{table}[!htbp]
	\centering
	\caption{Transition Datasets.}
	\label{tab:transition}
	\begin{tabular}{|c|c|c|c|c|}
		\hline
		\textbf{Dataset} &$\left|\ts\right|$ & Latitude & 
		Longitude  \\ \hline \hline
		{\em LA-Transit}  &  $109,036$ & [$\ang{32}$,$\ang{35}$] &[$\ang{-120}$,$\ang{-117}$] \\ \hline
		{\em NYC-Transit} &  $195,833$  & [$\ang{39}$,$\ang{42}$]& $[\ang{-75}$, $\ang{-72}$] \\ \hline
		{\em NYC-Synthetic} &  $10,000,000$  & [$\ang{39}$,$\ang{42}$]& $[\ang{-75}$, $\ang{-72}$] \\ \hline
	\end{tabular}
	\vspace{-2em}
\end{table}

\begin{table}[!htbp]
	\centering
	\caption{Parameter Settings.}
	\label{tab:parameter}
	\begin{tabular}{|c|c|}
		\hline
		\textbf{Para} & \textbf{Value}  \\
		\hline
		\hline
		$\left| Q \right|$ & $3$,$4$,\underline{$5$},$6$,$7$,$8$,$9$,$10$  \\
		\hline
		$k$ & $1$,{$5$},\underline{$10$},$15$,$20$,$25$  \\
		\hline
		$\mathcal{I}$ & $1km$, $2km$, \underline{$3km$}, $4km$, $5km$, $6km$\\
		\hline
		$\traveldistance{{se}}$ & $10km$, $20km$, \underline{$30km$}, $40km$, $50km$\\
		\hline
		$\frac{\tau}{\traveldistance{{se}}}$ & $1$, $1.2$, \underline{$1.4$}, $1.6$, $1.8$, $2.0$\\
		\hline
	\end{tabular}
	\vspace{-2em}
\end{table}

We conducted experiments to evaluate our solutions to {\rknnt} and
{\maxrknnt} using real bus route data and check-in data from
Foursquare\footnote{https://foursquare.com/} in New York and Los
Angeles, which are two largest cities in the USA.
We have published our
dataset\footnote{https://sites.google.com/site/shengwangcs/home/rknnt}
to improve the reproducibility of our results.
Figure~\ref{fig:heatmap} shows the heatmap of the route and check-in
datasets.
All experiments were performed on a machine using an Intel Xeon E$5$
CPU with $256$ GB RAM running RHEL v$6.3$ Linux, implemented in C++,
and compiled using GCC $4.8.1$ with -O2 optimization enabled.

\myparagraph{Route Datasets}We 
use two real bus network datasets, namely {\em NYC-Route} and
{\em LA-Route}.
We extracted the data from the GTFS datasets of New York
\footnote{http://web.mta.info/developers/developer-data-terms.html\#data} 
and Los Angeles\footnote{http://developer.metro.net/gtfs/google\_transit.zip}.
Table~\ref{tab:route} provides a breakdown of each dataset.

\myparagraph{Transition Datasets}Two 
real transition datasets, {\em NYC-Transit} and {\em LA-Transit},
were produced by cleaning the Foursquare check-in
data~{\cite{Bao2012Location}}, and statistics for the cleaned data is
shown in Table~\ref{tab:transition}.
Specifically, we divided a user's trajectory with multiple points
into several transitions with two points.
A trajectory with $n$ points can be divided into $n-1$ transitions.
Since the real dataset is small, we also generated a synthetic
dataset which contains $10$ million transitions for the NYC dataset,
and refer to it as {\em NYC-Synthetic}.

\subsection{Evaluation of {\large \rknnt}}
\myparagraph{Algorithms for evaluation}
We compared the following methods when processing {\rknnt} over the
two datasets.
(1) \textbf{Filter-Refine}: The basic framework proposed in
Section~{\ref{sec:baseline}}.
(2) \textbf{Voronoi}: The Voronoi-based method which can create a
larger filtering area by drawing a Voronoi diagram based on the query
and filtering route after regular filtering by points.
(3) \textbf{Divide-Conquer}: As proposed in Section~{\ref{sec:dc}}.

\myparagraph{Queries}
We prepared two query sets: the first set is a synthetic query set
for the purposes of parameter evaluation, and generated as follows: 
1) We randomly generated $1,000$ points from $\rs$.
2) We iteratively chose each point as a start point, and append new
points one by one with a limited rotation angle to simulate a
realistic case.
The rotation angle of every time extension does not exceed
$\ang{90}$, so the query route will not zigzag~{\cite{Chen2014a}}.
All experimental results are averaged by running all $1,000$ queries.
The second query set contains all the routes in {\em NYC-Transit} and
{\em LA-Transit}, which are used as queries to test our most
efficient method, {\comparethree}.

\myparagraph{Parameters}
Table~\ref{tab:parameter} summarizes all key parameters for a
query, and the default values are underlined.
$\mathcal{I} = \frac{\traveldistance{Q}}{|Q|}$ is the interval length
between two adjacent points in the query, where $\traveldistance{Q}$
is the travel distance of the query route and can be computed by
Equation~\ref{equ:td}.

\begin{figure}[h]
	\centering
	\vspace{-1em}
	\subfigure[\textbf{LA}]
	{ \label{fig:ktimela}
		\includegraphics[width=4cm]{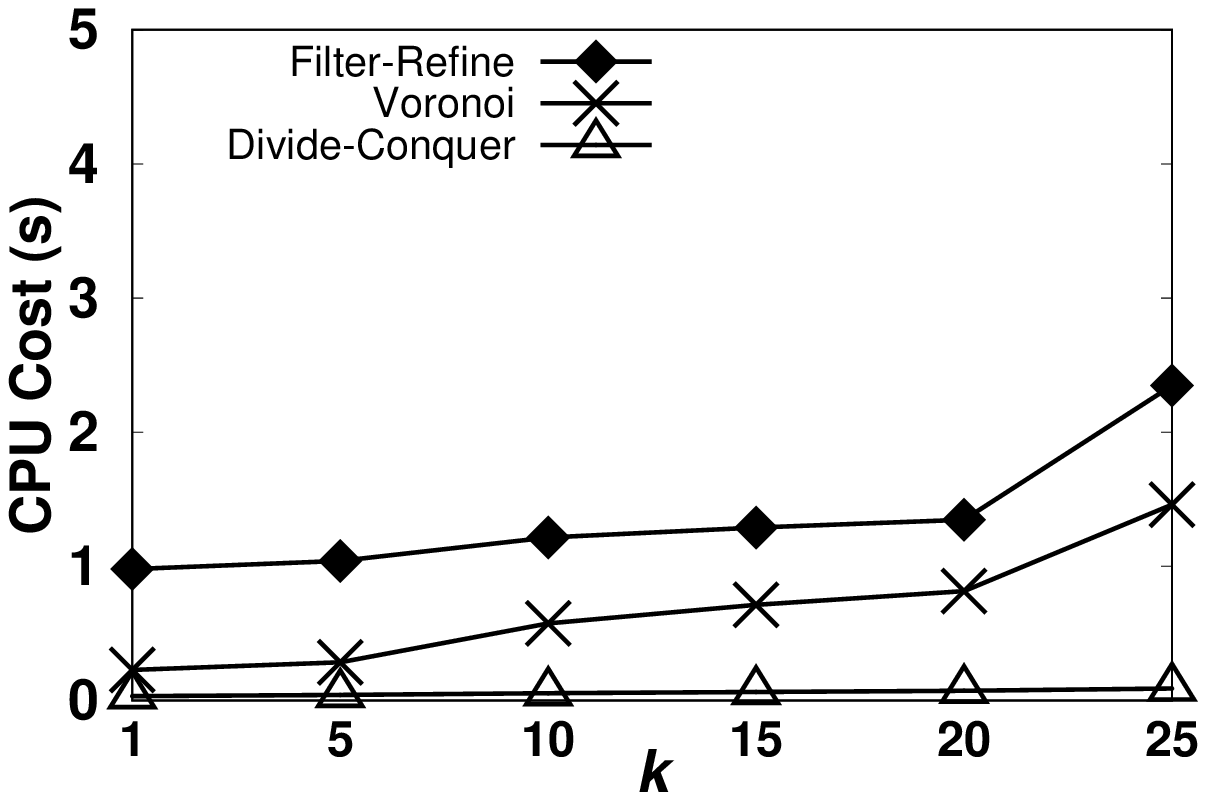}
	}
	\subfigure[\textbf{NYC}]
	{ \label{fig:ktimenyc}
		\includegraphics[width=4cm]{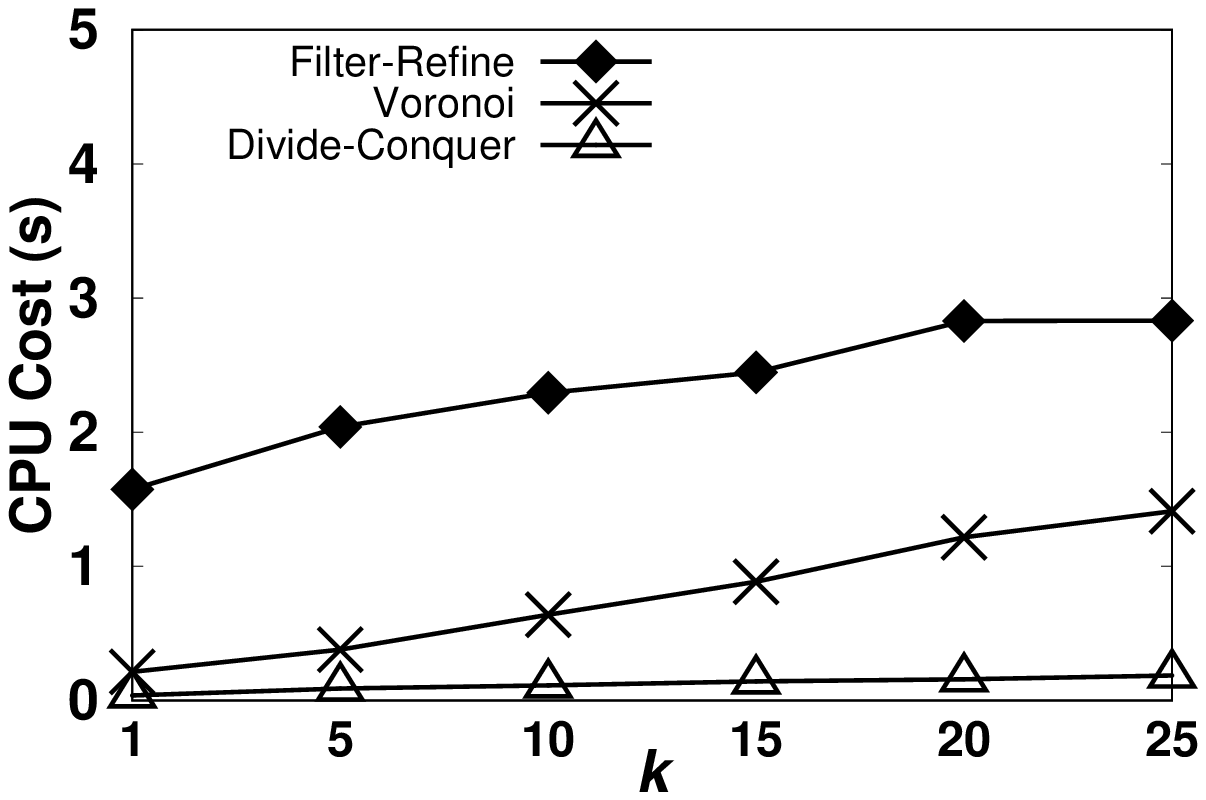}
	}
	\vspace{-2em}
	\caption{Effect on Running Time with the increasing of $k$.}
	\label{fig:ktime}
	\vspace{-1em}
\end{figure}

\begin{figure}[h]
	\centering
	{
		\includegraphics[width=8cm]{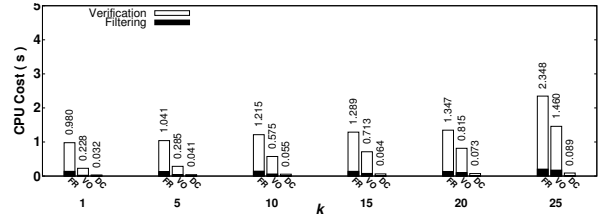}
	}
	\vspace{-1em}
	\caption{Breakdown of running time with increasing $k$ in \textbf{LA}.}
	\label{fig:columnkla}
	\vspace{-1em}
\end{figure}

\begin{figure}[h]
	\centering
	\subfigure[\textbf{LA}]
	{ \label{fig:qtimela}
		\includegraphics[width=4cm]{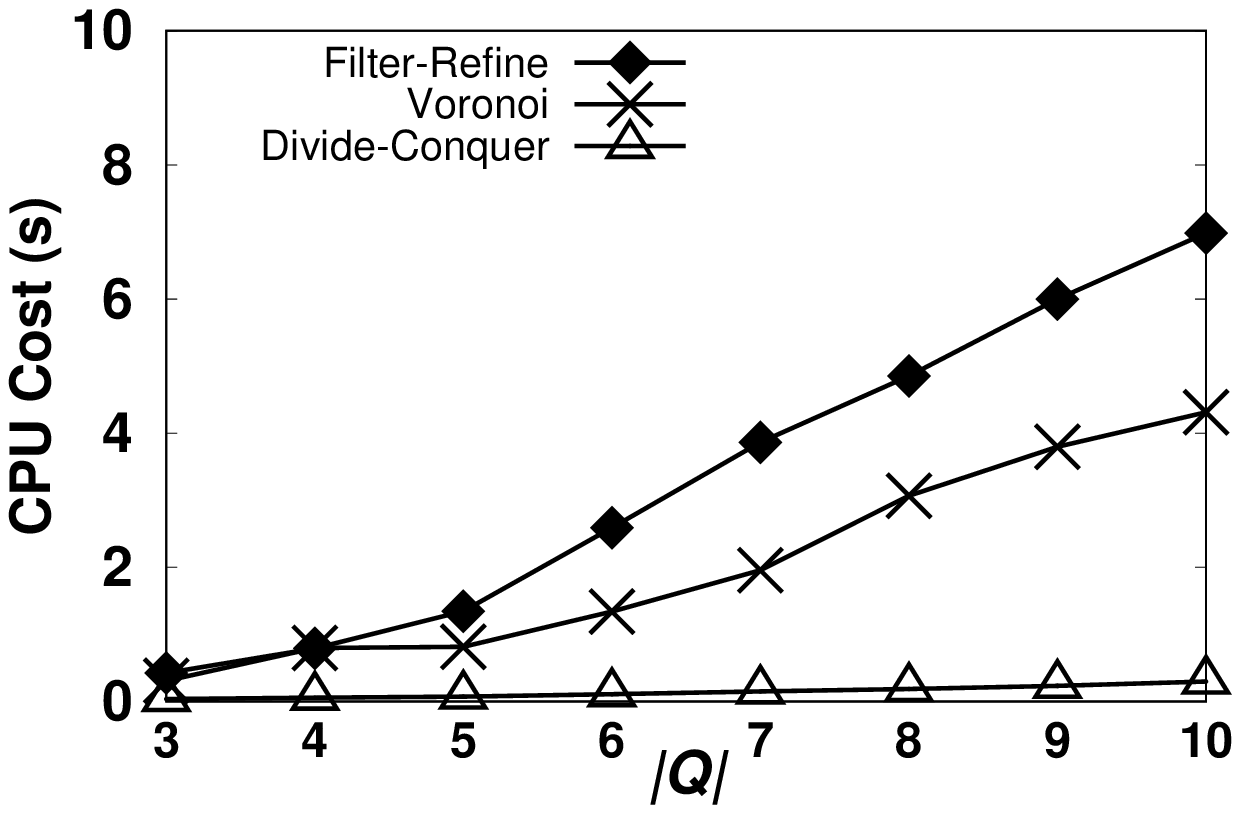}
	}
	\subfigure[\textbf{NYC}]
	{ \label{fig:qtimenyc}
		\includegraphics[width=4cm]{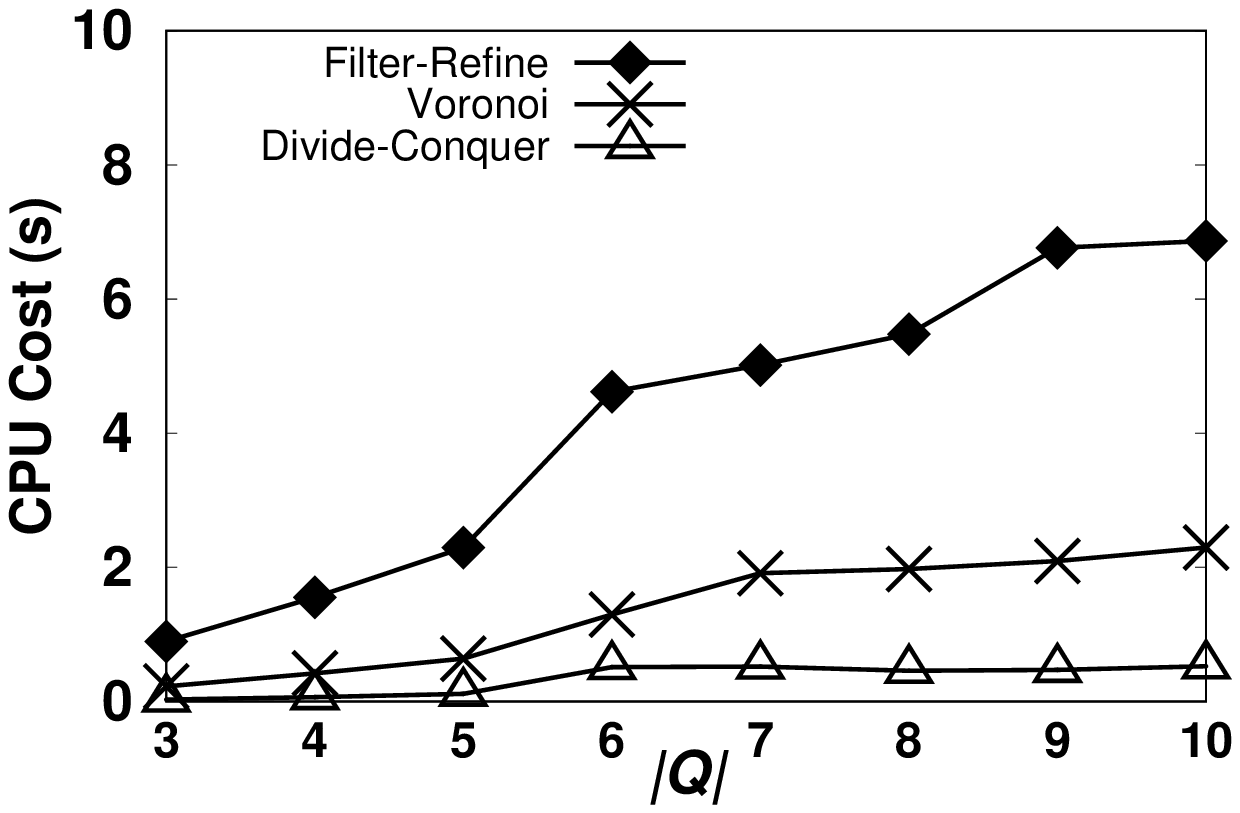}
	}
	\vspace{-2em}
	\caption{Effect on running time with the increasing of $|Q|$.}
	\label{fig:qtime}
	\vspace{-1em}
\end{figure}

\begin{figure}[h]
	\centering
	{
		\includegraphics[width=8cm]{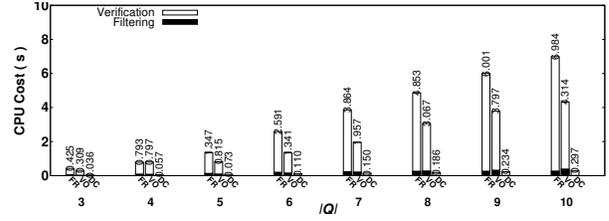}
	}
	\vspace{-1em}
	\caption{Breakdown of running time w.r.t. $|Q|$ in \textbf{LA}.}
	\label{fig:columnqla}
	\vspace{-1em}
\end{figure}

\begin{figure}[h]
	\centering
	\subfigure[\textbf{Synthetic}]
	{ \label{fig:ktimesyn}
		\includegraphics[width=4cm]{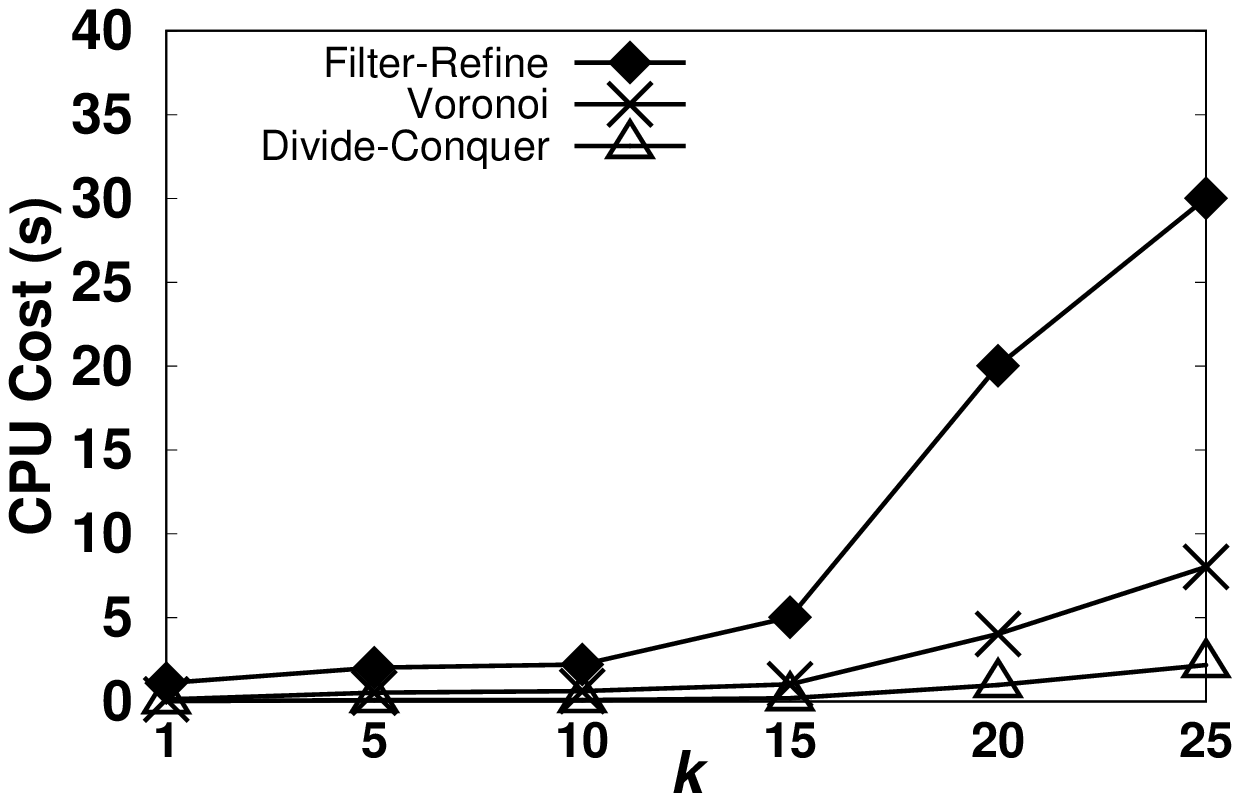}
	}
	\subfigure[\textbf{Synthetic}]
	{ \label{fig:qtimesyn}
		\includegraphics[width=4cm]{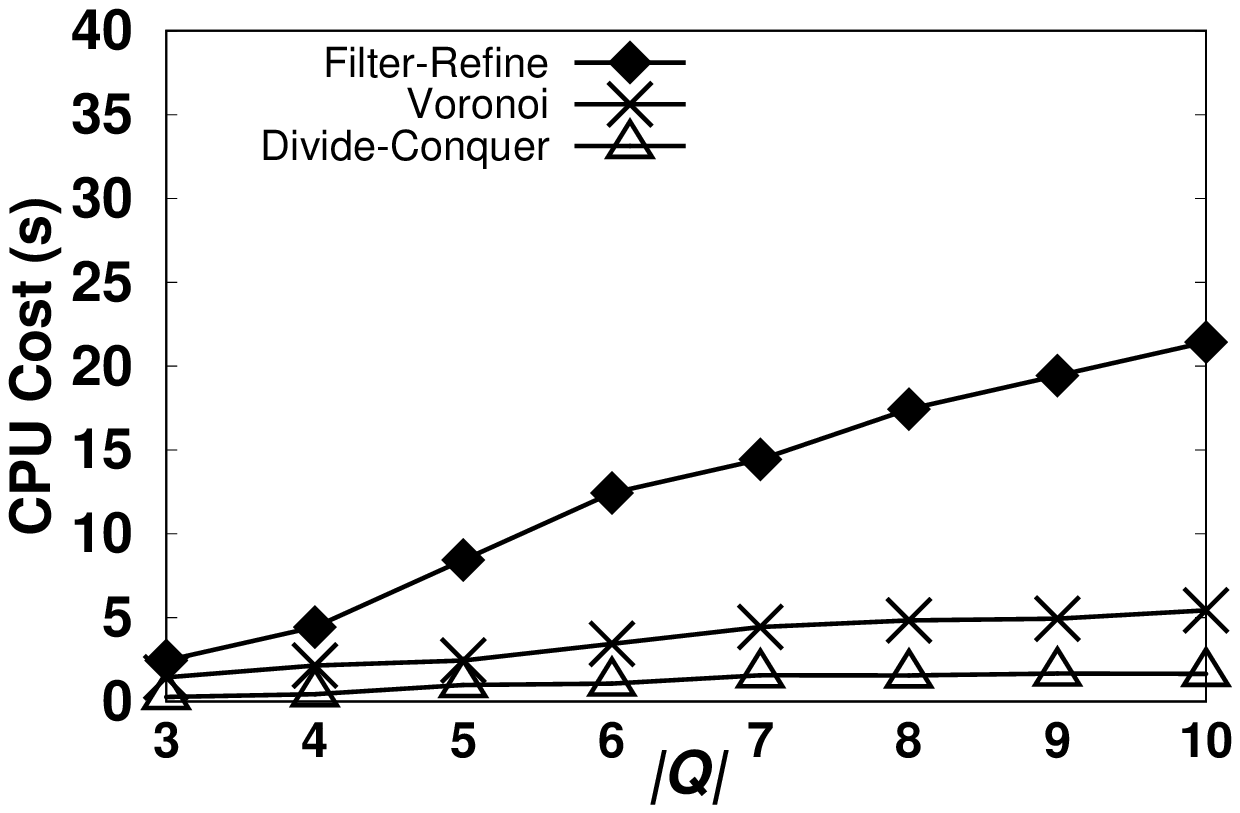}
	}
	\vspace{-2em}
	\caption{Effect on running time with the increasing of $|Q|$ and $k$ in synthetic dataset.}
	\label{fig:synthetictime}
	\vspace{-1em}
\end{figure}

\begin{figure}[h]
	\centering
	\subfigure[\textbf{LA}]
	{ \label{fig:itimela}
		\includegraphics[width=4cm]{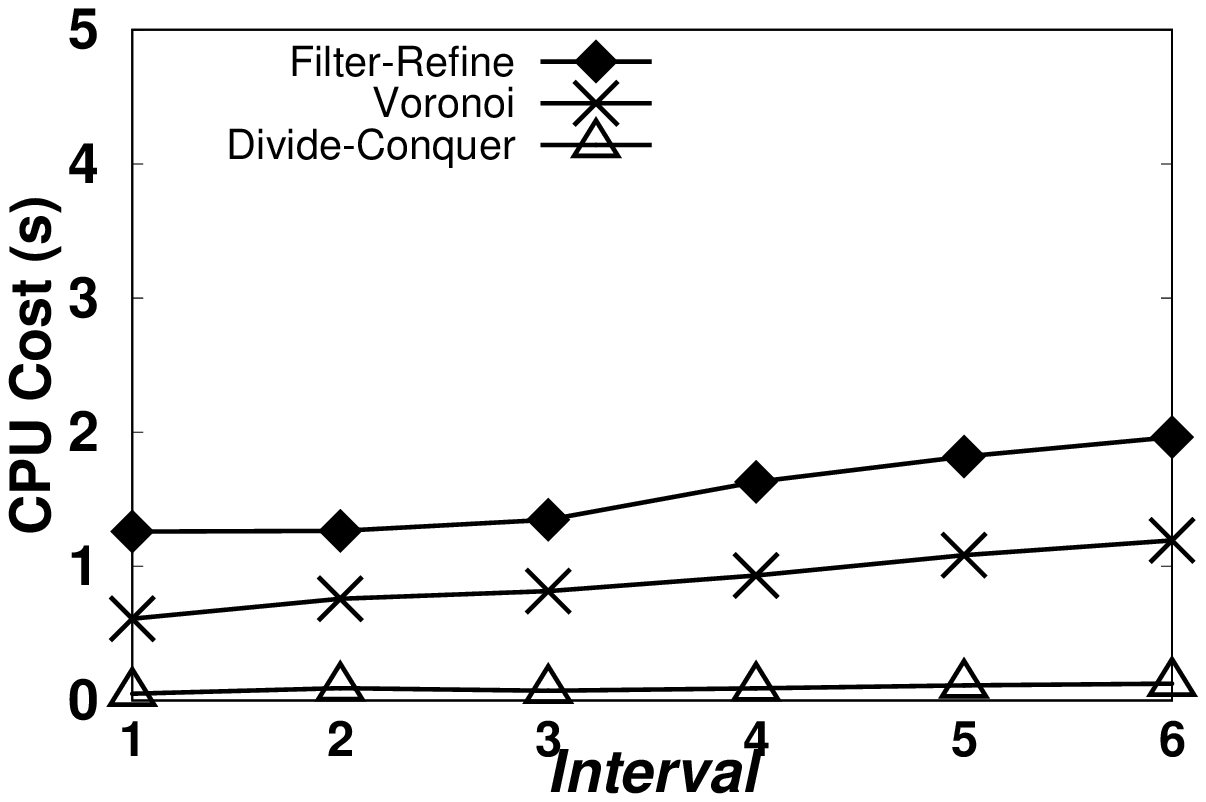}
	}
	\subfigure[\textbf{NYC}]
	{ \label{fig:itimenyc}
		\includegraphics[width=4cm]{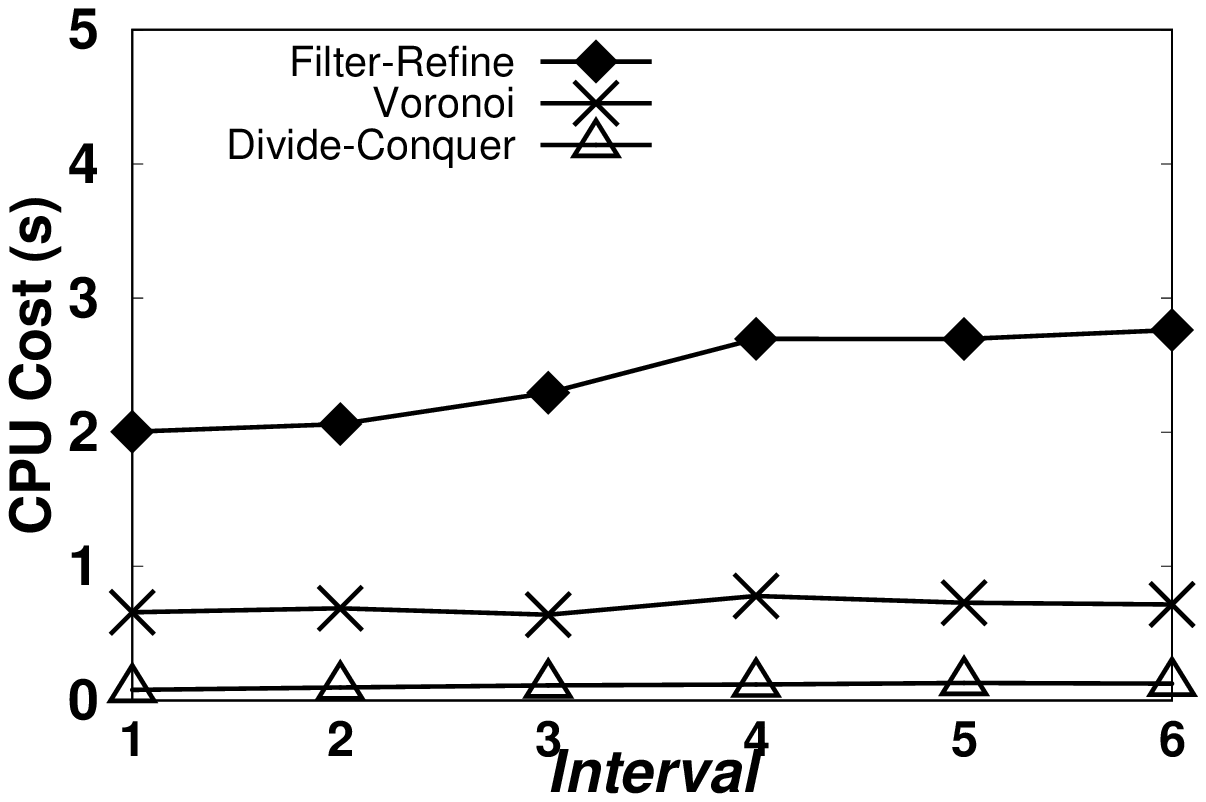}
	}
	\vspace{-2em}
	\caption{Effect on running time with the increasing of $\mathcal{I}$.}
	\label{fig:itime}
	\vspace{-1em}
\end{figure}

\begin{figure}[h]
	\centering
	{
		\includegraphics[width=8cm]{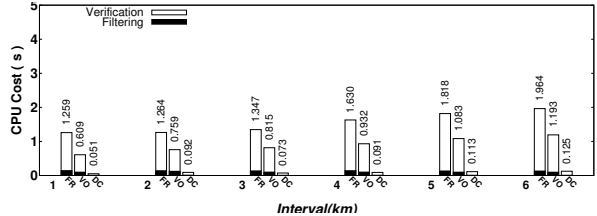}
	}
	\vspace{-1em}
	\caption{Breakdown of running time with increasing $\mathcal{I}$ in \textbf{LA}.}
	\label{fig:columnila}
	\vspace{-1em}
\end{figure}

\begin{figure}[h]
	\centering
	\subfigure[\textbf{LA}]
	{ \label{fig:rknntla}
		\includegraphics[width=4cm]{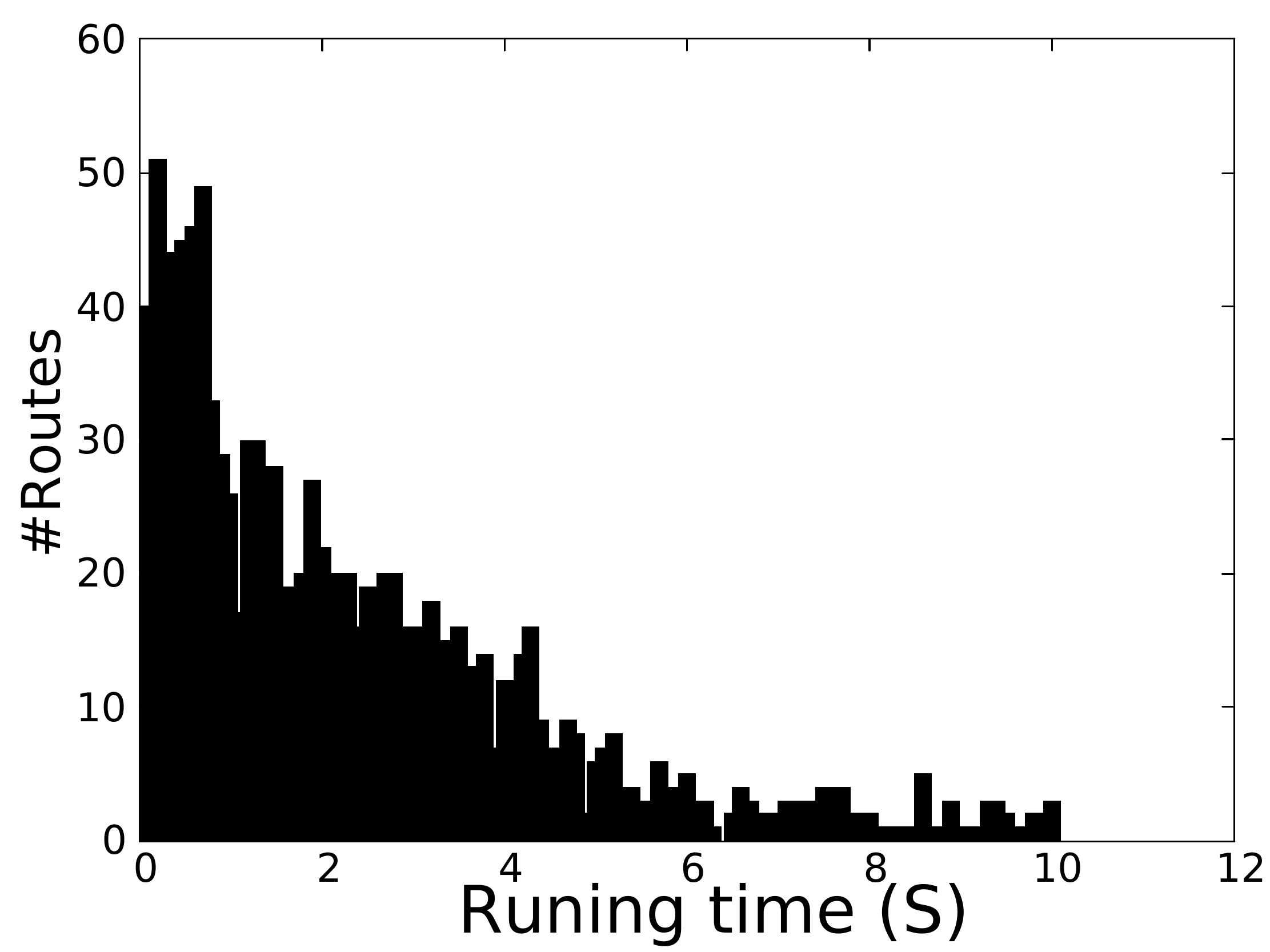}
	}
	\subfigure[\textbf{NYC}]
	{ \label{fig:rknntnyc}
		\includegraphics[width=4cm]{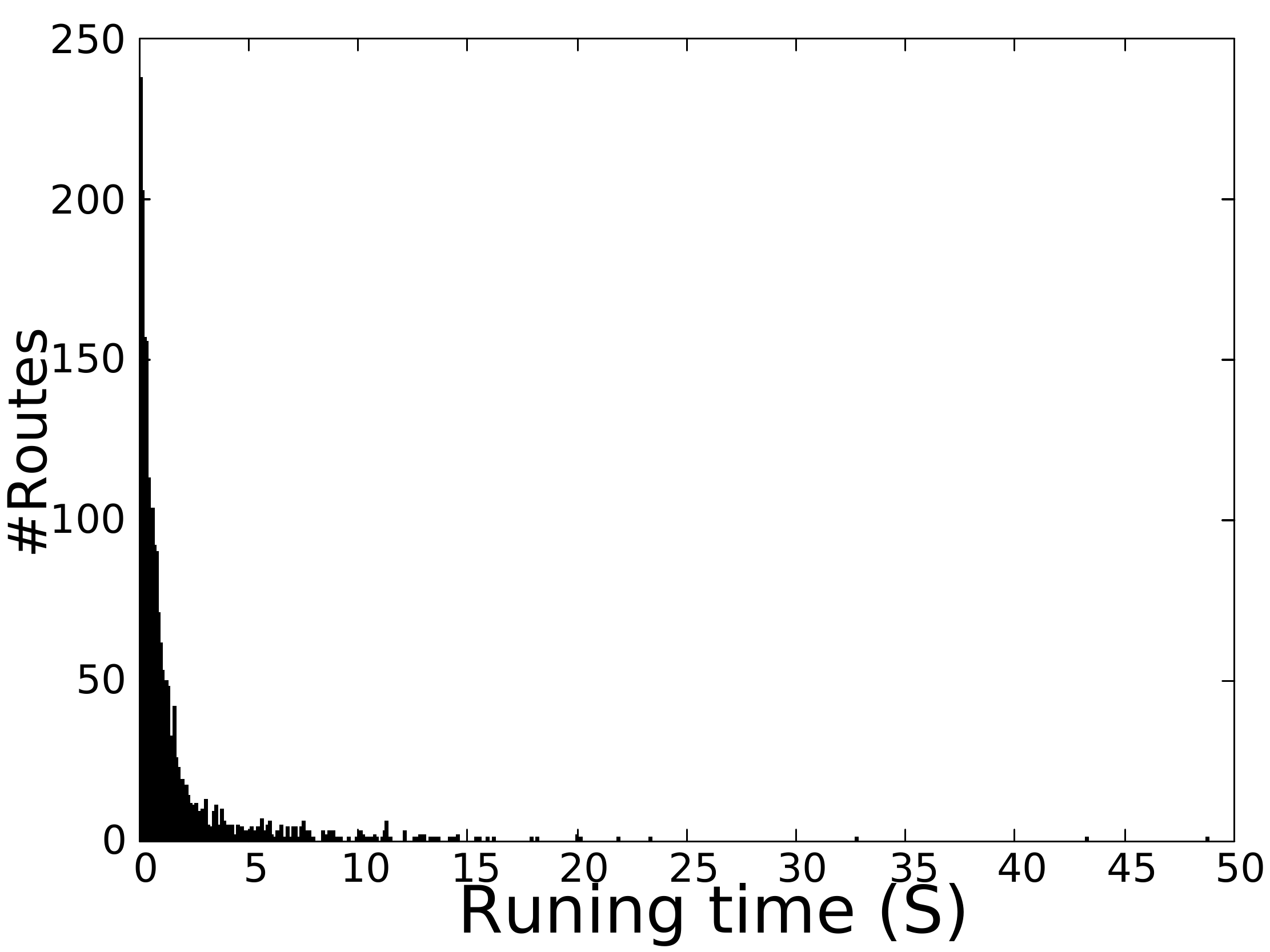}
	}
	\vspace{-2em}
	\caption{The distribution of running time when taking all existing bus routes as query by \maxrknnt~when $k=10$.}
	\label{fig:rknnt}
	\vspace{-1em}
\end{figure}

\myparagraph{Effect of $\left|Q\right|$}
Figure~\ref{fig:qtime} shows the running time of our three methods.
As more points are added into the query, {\compareone} and
{\comparetwo} exhibit a sharp increase in running time.
Since these methods need more time to check whether a node can be
filtered, the filtering space becomes smaller and the probability of
being pruned decreases.
In contrast, {\comparethree} shows almost a linear increase.
This is probably a result of the whole query being divided into $|Q|$
queries, and a node is not be pruned by checking every query point.
Figure~\ref{fig:columnqla} shows a breakdown of the running time to
the tasks of filtering and verification on the LA data.
We can see that the verification occupies more than $80\%$ for most
cases.

\myparagraph{Effect of $k$}
Figure~\ref{fig:ktime} shows that the time cost for all three methods
will increases as $k$ increases.
This is because it is unlikely that a point can be filtered by $k$
filtering routes when $k$ is large.

\myparagraph{Effect of $\mathcal{I}$}
We observe that the intervals $\mathcal{I}$ between two adjacent
points vary from route to route in real life.
Hence, we conducted experiments to see how the running
time is affected in this scenario.
As mentioned when describing query generation, the size of the query
is increased by appending randomly generated points, one at a time.
Figure~\ref{fig:itimela} and Figure~{\ref{fig:itimenyc}} show that
there is a slight increase on the running time when $\mathcal{I}$ is
large.
The main reason is that when two query points are close, a node can
be filtered by a filtering point easily, while when the intervals
are large, it is harder to filter a node.

\myparagraph{Real Route Queries}
After testing the effect of each individual parameter, we took every
route in each dataset as a query to evaluate our best
method $\comparethree$.
Note that before running each query, we removed the points of this
route from the $\fr$ index.
Figure~\ref{fig:rknnt} shows that over $90\%$ of the queries can be
processed in less than $5$s.
The main reason is the relationship to the number of points in
the query.

In summary, our main observations are:
\begin{enumerate}[leftmargin=1em]
	\itemsep 0em
	\item {\comparethree} consistently has the best performance, followed
	by {\comparetwo}, with {\compareone} being the worst.
	\item All three methods are sensitive to $k$ and $|Q|$.
	Only {\compareone} and {\comparetwo} are sensitive to the
	interval length $\mathcal{I}$ of the query.
	\item When taking existing routes as real queries, most queries can
	be answered efficiently using {\comparethree}.
\end{enumerate}

\subsection{Evaluation of  {\large \maxrknnt}}
\label{sec:em}
\myparagraph{Algorithms for evaluation}(1) \textbf{BruteForce}: 
the baseline method which uses the {\em k shortest
	paths}~{\cite{Yen1971}} to find all the routes which have a smaller
travel distance than the distance threshold $\tau$, after which an
{\rknnt} is performed on every candidate to choose the maximal one.
(2) \textbf{Pre}: the method that extends the \textbf{BrouteForce} by
pre-computation of the {\rknnt} set for every vertex without an
on-the-fly {\rknnt} query.
For {\maxrknnt} and {\minrknnt}, both can be solved using the same
pruning techniques with little difference in bound checking, which
has a small impact on performance. We denote them as (3)
\textbf{Pre-Max} and (4) \textbf{Pre-Min}.

\myparagraph{Queries}
To test the effect of key parameters, we first generated a point set
by choosing $1,000$ start points randomly from our route datasets.
Then, we searched $6$ end points for every start point with different
$\traveldistance{se}$, which is the distance between the origin and
the destination, as shown in the last row of Table~\ref{tab:parameter}.
Furthermore, we used existing representative routes as queries and
employed  {\maxrknnt} and {\minrknnt} search algorithm to find the
new ``optimal'' routes.
Finally, we compared the {\rknnt} sets of the original routes against
the new routes.

\myparagraph{Parameters}
We discovered two key parameters that affect the performance of
{\maxrknnt}: (1) the coverage degree of a bus route - denoted by
$\traveldistance{{se}}$ and quantified as the Euclidean distance
between the start and end points of a query $Q$.
(2) $\frac{\tau}{\traveldistance{{se}}}$, which is the ratio of the
travel distance over the straight-line distance from origin to
destination of $Q$.
The choices of these parameters are from the distribution of all real
bus routes, as shown in Figure~\ref{fig:frequence-la}.

\myparagraph{Pre-computation}
Table~\ref{tab:precomputation} shows the time spent on
pre-computation.
The pre-computation consists of of two steps: the {\rknnt} query
for every vertex, and the shortest distance route search, as shown
in Algorithm~\ref{alg:pre}.
All-pair shortest distance computation costs about $4$ minutes for
both datasets, and the {\rknnt} search of all vertices in $\graph$ costs
less than $5$ minutes when $k=10$.
For the synthetic dataset which contains 10M transitions, the time spent
on pre-computation is about $12$ minutes when $k=10$. 

\begin{table}
	\centering
	\vspace{-2.5em}
	\caption{Running time(s) for pre-computation when
		$k=1,5,10$, which is composed of {\rknnt} search and all-pair
		shortest distance computations, the bold numbers are the results for 
		synthetic dataset.}
	\label{tab:precomputation}
	\begin{tabular}{|c|c|c|c|c|c|c|}
		\hline
		& \multicolumn{3}{c|}{\textbf{LA}}                                                       & \multicolumn{3}{c|}{\textbf{NYC}}                                              \\ \hline
		$k$                     & $1$                     & $5$                      & $10$                     & $1$                   & $5$                   & $10$                  \\ \hline
		\multirow{2}{*}{\rknnt} & \multirow{2}{*}{$80.5$} & \multirow{2}{*}{$153.2$} & \multirow{2}{*}{$230.8$} & $140.4$               & $202.1$               & $253.5$               \\ \cline{5-7} 
		&                         &                          &                         & \multicolumn{1}{l|}{\textbf{201.7}} & \multicolumn{1}{l|}{\textbf{545.8}} & \multicolumn{1}{l|}{\textbf{748.1} } \\ \hline
		\textbf{Shortest}     & \multicolumn{3}{c|}{$191.3$}                                                  & \multicolumn{3}{c|}{$251.9$}                                          \\ \hline
	\end{tabular}
\end{table}

\begin{figure}
	\centering
	{
		\includegraphics[height=3cm]{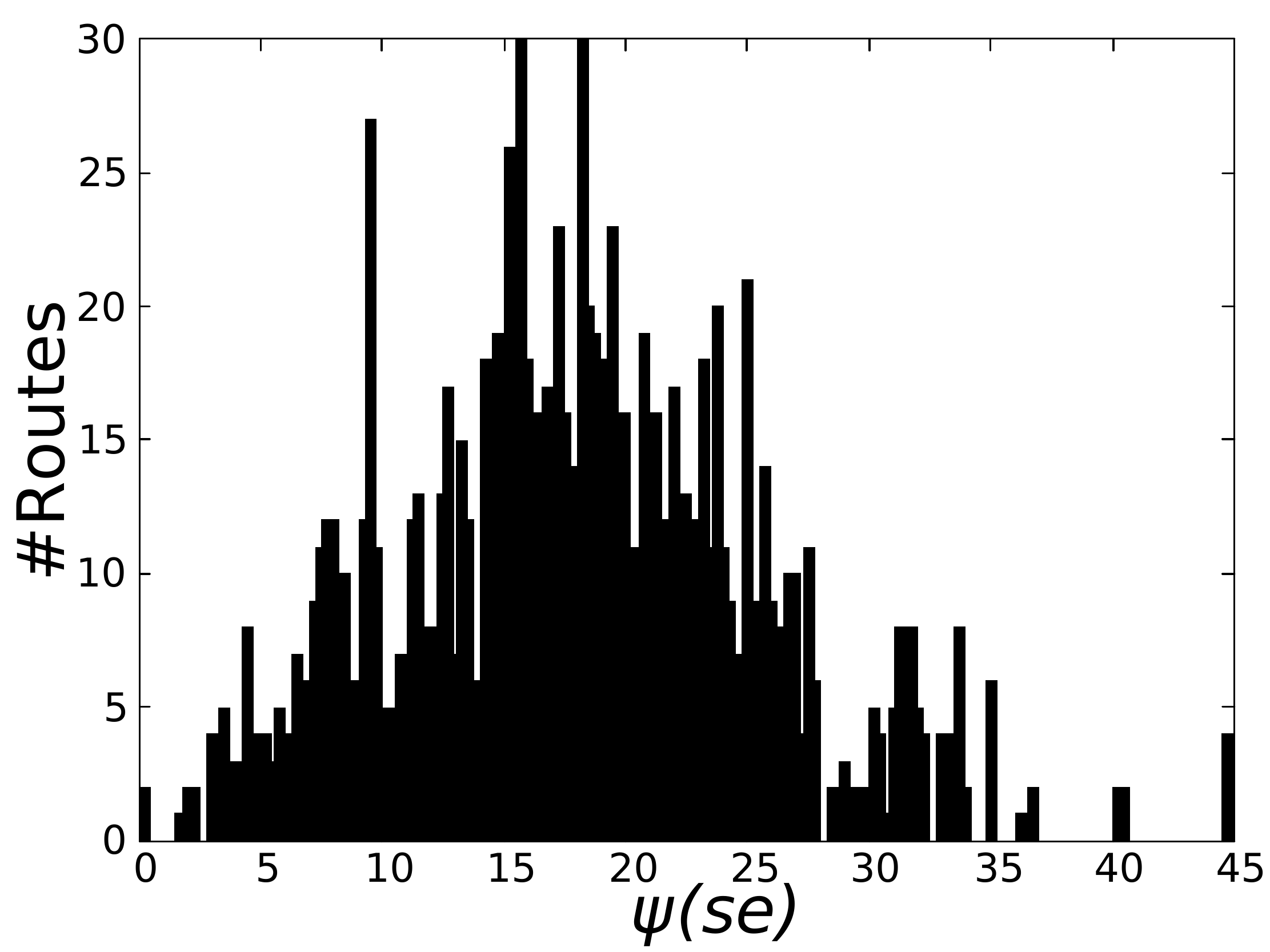}
		\includegraphics[height=3cm]{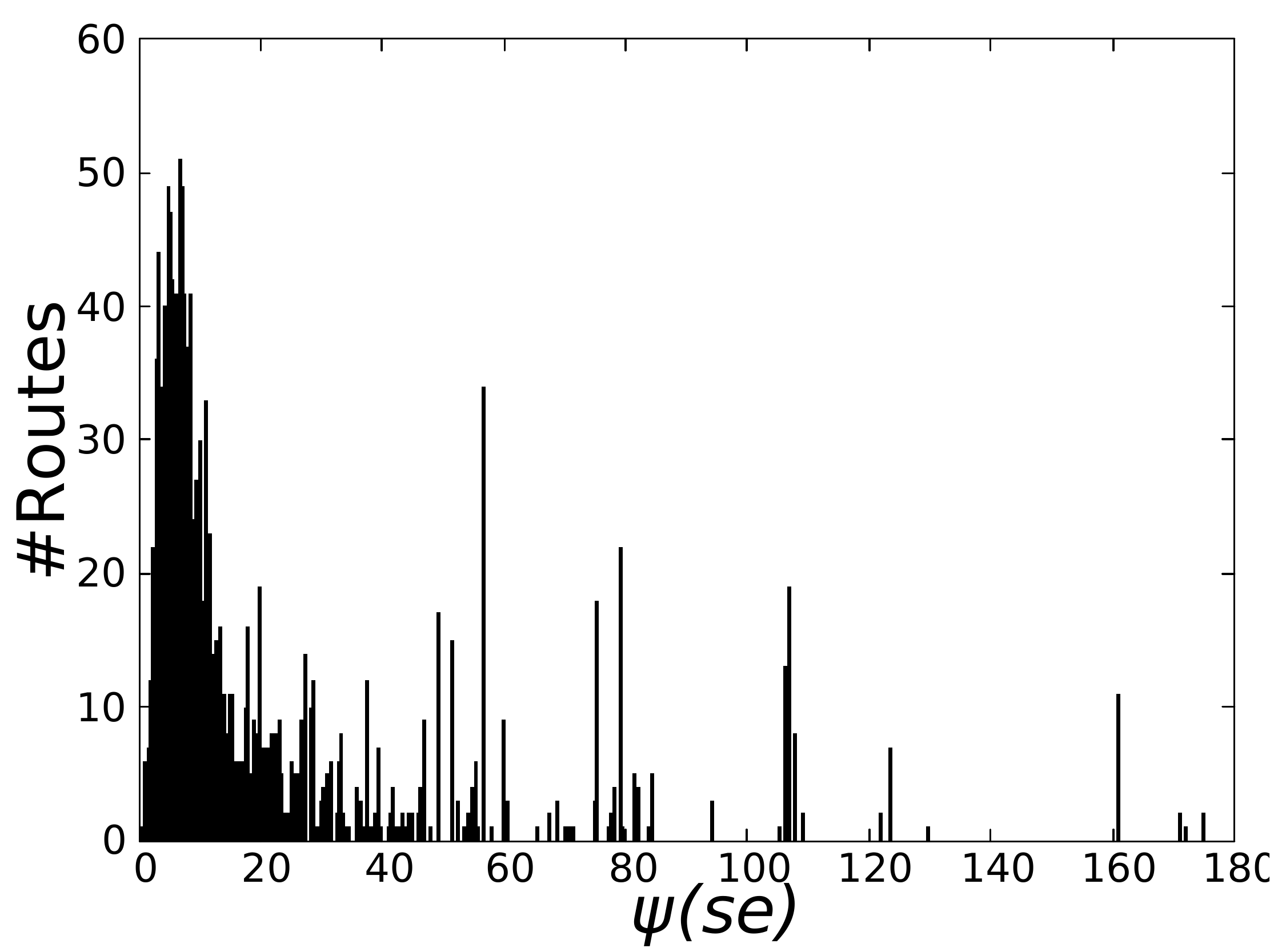}\\~
		\includegraphics[height=3cm]{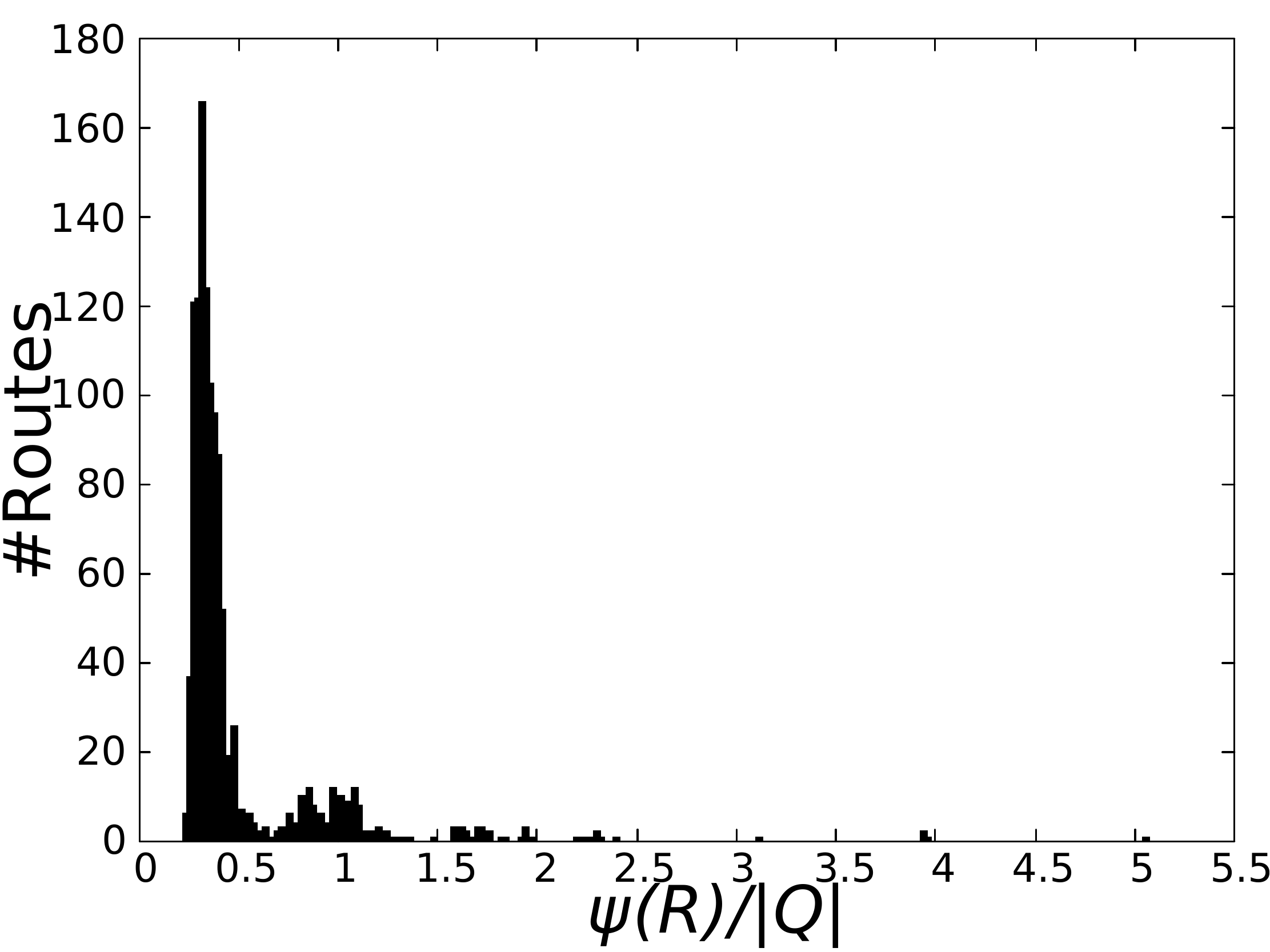}
		\includegraphics[height=3cm]{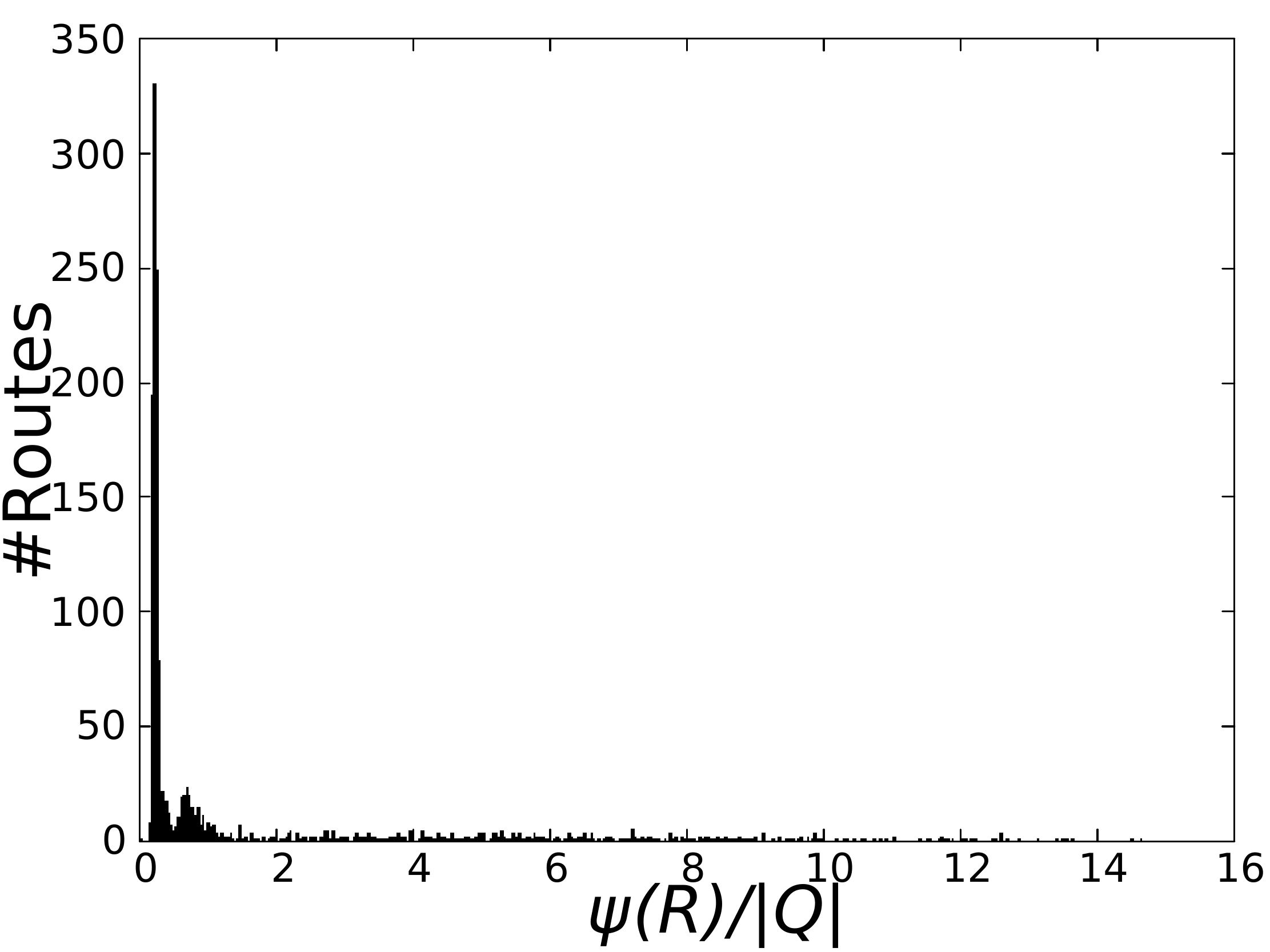}\\~~~
		\includegraphics[height=3cm]{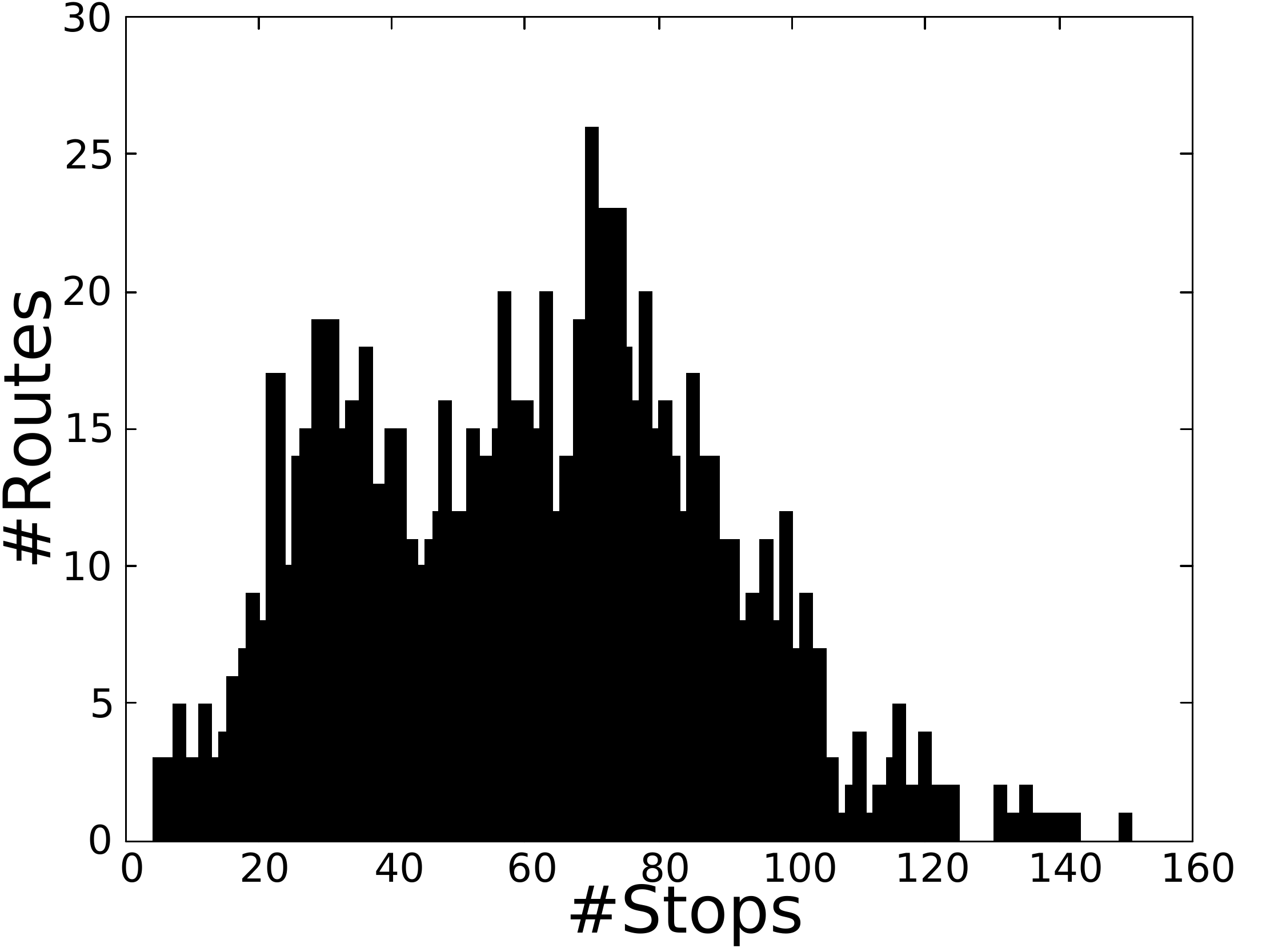}
		\includegraphics[height=3cm]{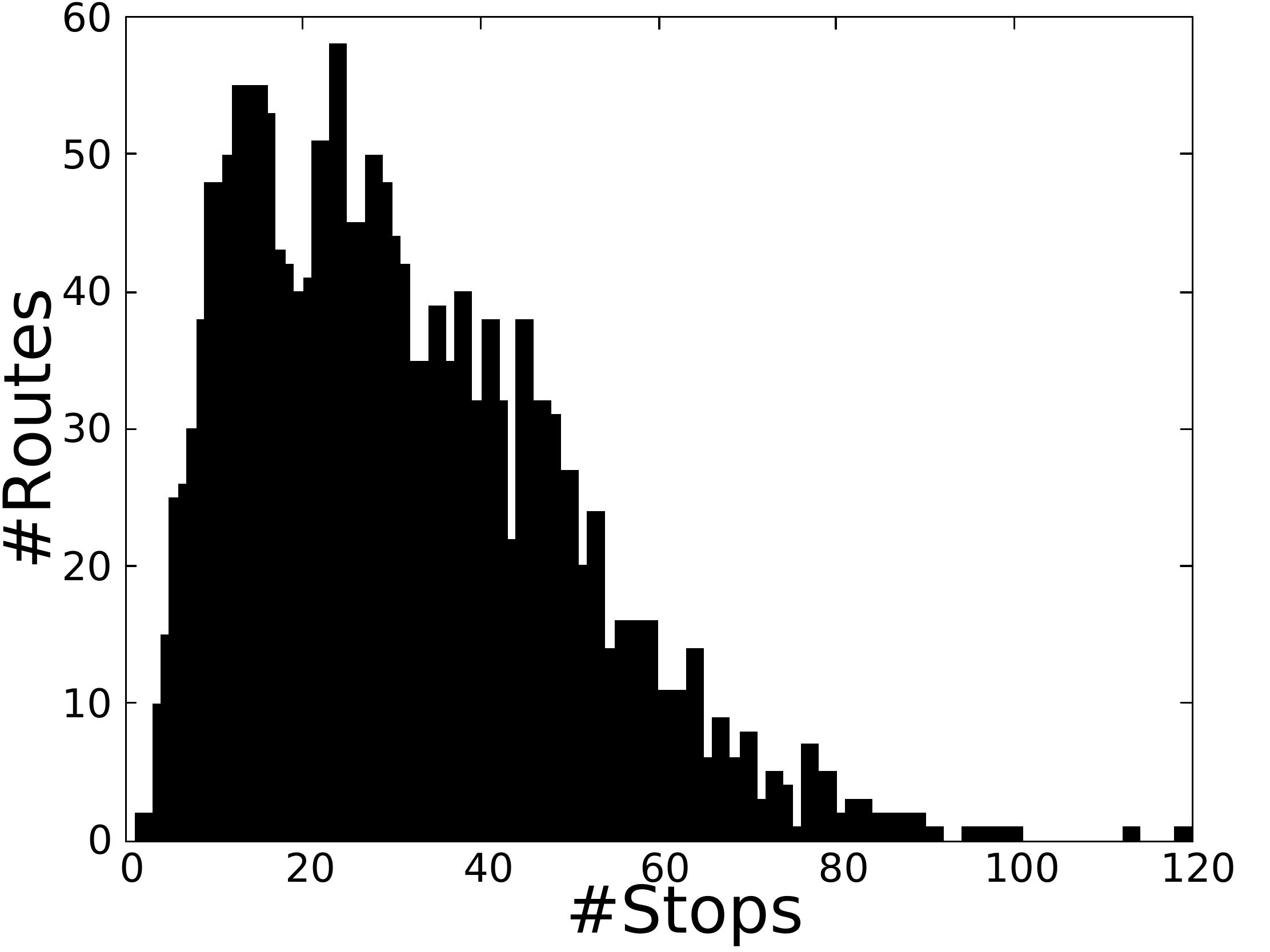}
	}
	\vspace{-2em}
	\caption{Frequency histogram of $\traveldistance{{se}}$, $\mathcal{I}$ and $|\route|$ in LA (left) and NYC (right).}
	\label{fig:frequence-la}
	\vspace{-1em}
\end{figure}

\myparagraph{Effect of $\traveldistance{{se}}$}
Figure~\ref{fig:odtimela} shows that the time spent on the search
task increases when the distance between the origin and destination
$\traveldistance{{se}}$ increases.
This is because more vertices in the graph need to be scanned between
the origin and destination.
For {\maxcompareone}, the reasons are twofold: (1) It returns more
candidate routes for {\rknnt}; (2) The candidate routes
are longer when $\traveldistance{{se}}$ is long, so more time has
to be spent for every $\rknnt$ query.
In contrast, for the remaining three methods, since we have
pre-computed the {\rknnt} set for every vertex, the running
time comes from the search over $\graph$.
\textbf{Pre-Max} has the best performance due to the bound checking
during the spreading of partial routes.

\begin{figure}[h]
	\centering
	\subfigure[\textbf{LA}]
	{ \label{fig:odtimela}
		\includegraphics[width=4cm]{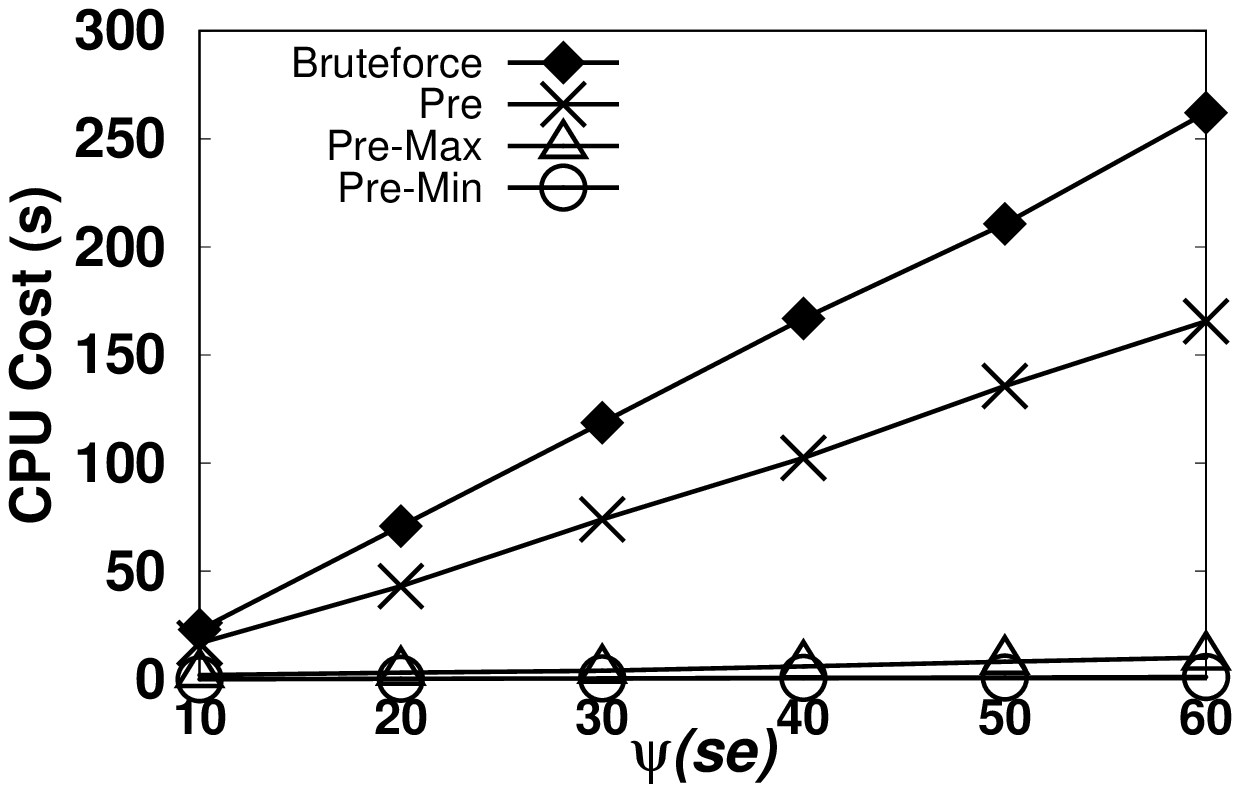}
	}
	\subfigure[\textbf{NYC}]
	{ \label{fig:odtimenyc}
		\includegraphics[width=4cm]{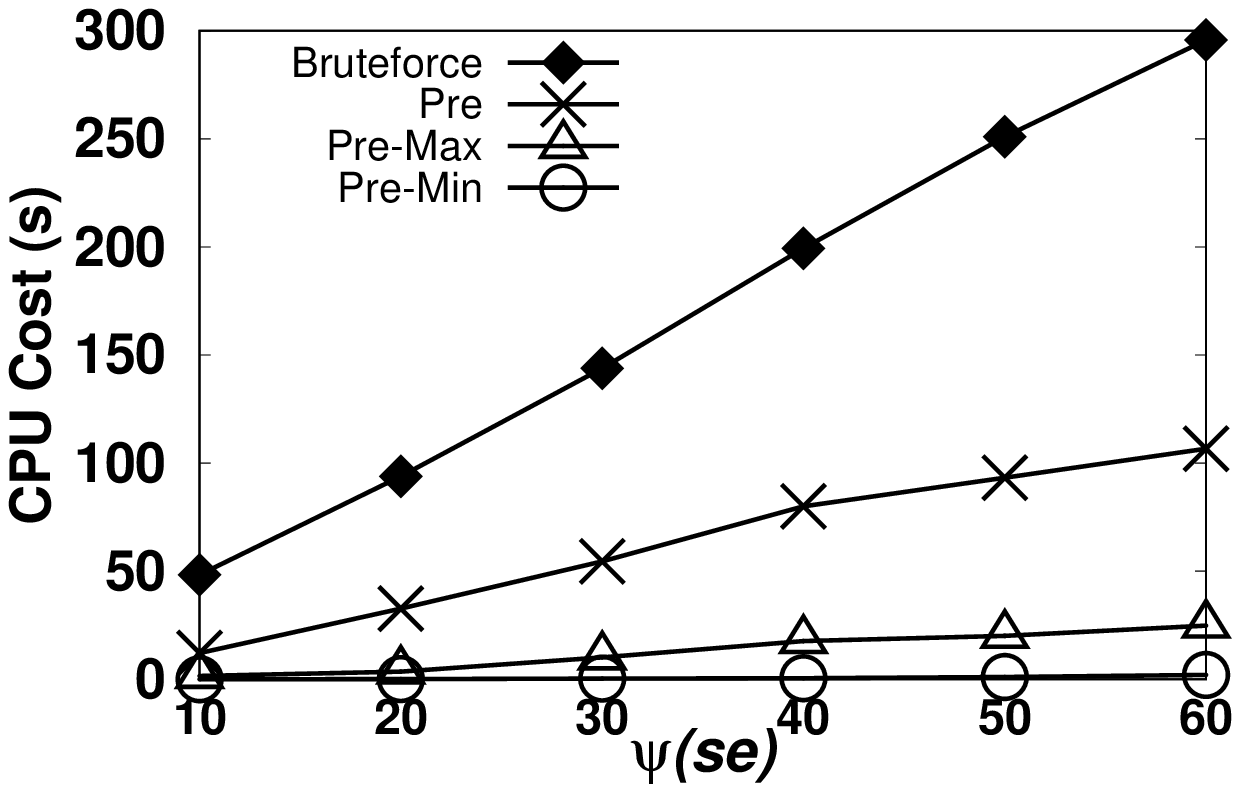}
	}
	\vspace{-2em}
	\caption{Effect on running time as $\traveldistance{{se}}$ increases.}
	\label{fig:oidtime}
	\vspace{-1em}
\end{figure}

\myparagraph{Effect of $\frac{\tau}{\traveldistance{{se}}}$}
To generate the query, we choose a subset of queries with a
fixed $\traveldistance{{se}}$ as the default value shown in
Table~\ref{tab:parameter} and alter $\tau$ in the experiment.
Figure~\ref{fig:dtime} shows that increasing
$\frac{\tau}{\traveldistance{{se}}}$ leads to an increased
running time.
The reason can also be ascribed to the increasing number of
candidates between the origin and destination.

\begin{figure}[h]
	\centering
	\subfigure[\textbf{LA}]
	{ \label{fig:dtimela}
		\includegraphics[width=4cm]{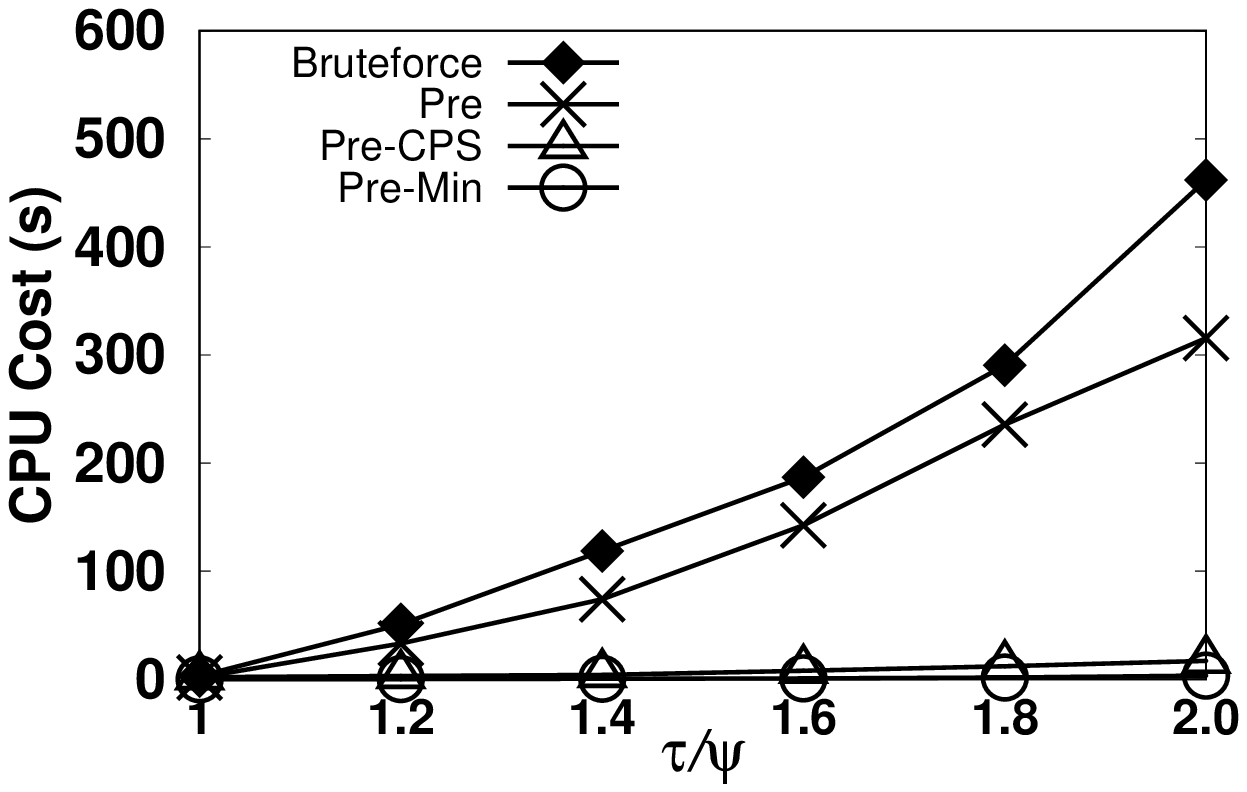}
	}
	\subfigure[\textbf{NYC}]
	{ \label{fig:dtimenyc}
		\includegraphics[width=4cm]{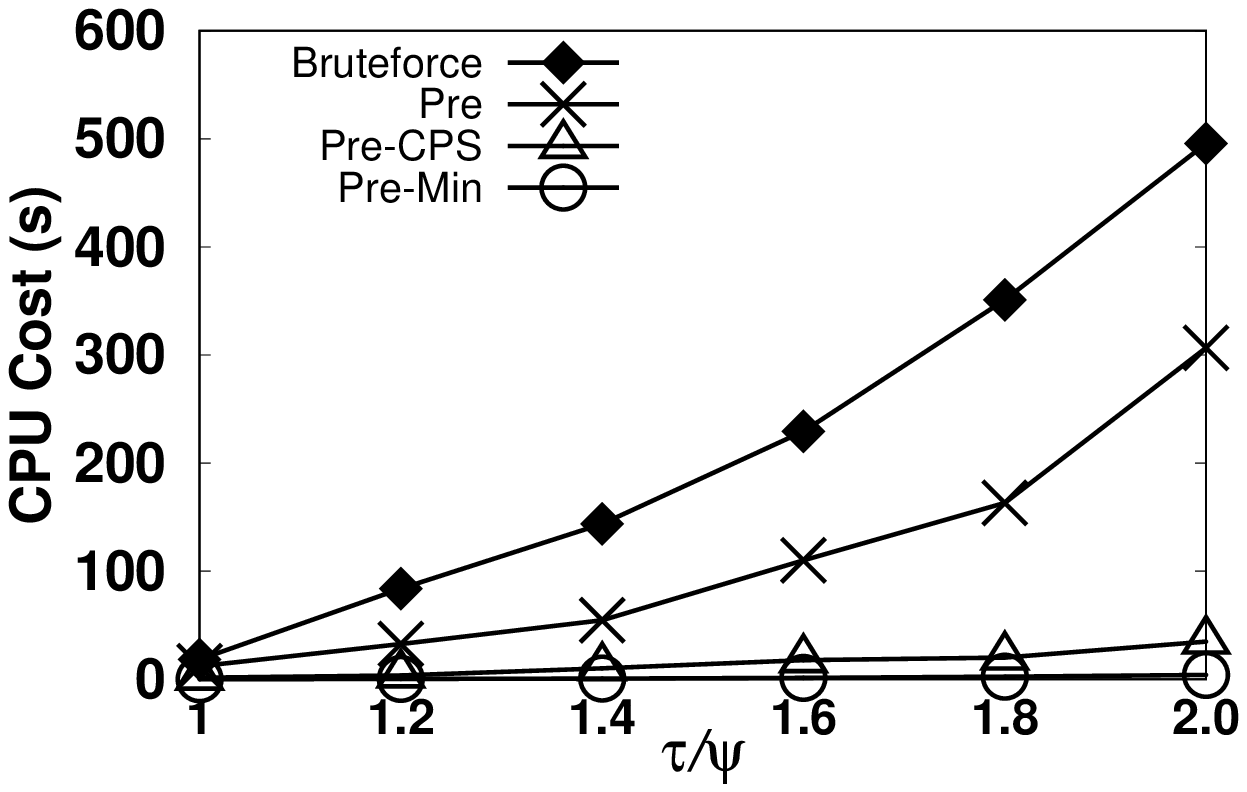}
	}
	\vspace{-2em}
	\caption{Effect on running time with the increase of $\frac{\tau}{\traveldistance{{se}}}$.}
	\label{fig:dtime}
	\vspace{-1em}
\end{figure}

\myparagraph{Real queries}
We took each route in $\rs$ as a query to perform an {\maxrknnt} search to
see whether we can find a better route which has a larger {\rknnt} set
while maintaining an acceptable travel distance threshold.
Each query is generated using the start and end bus stop, and the
travel distance for each route.
Figure~\ref{fig:maxrknnttime} shows the running time distribution for 
the real queries.
We can see that most queries in the LA data can be answered in less
than a second.

\begin{figure}[h]
	\centering
	\subfigure[\textbf{LA}]
	{ \label{fig:maxrknnttimela}
		\includegraphics[width=4cm]{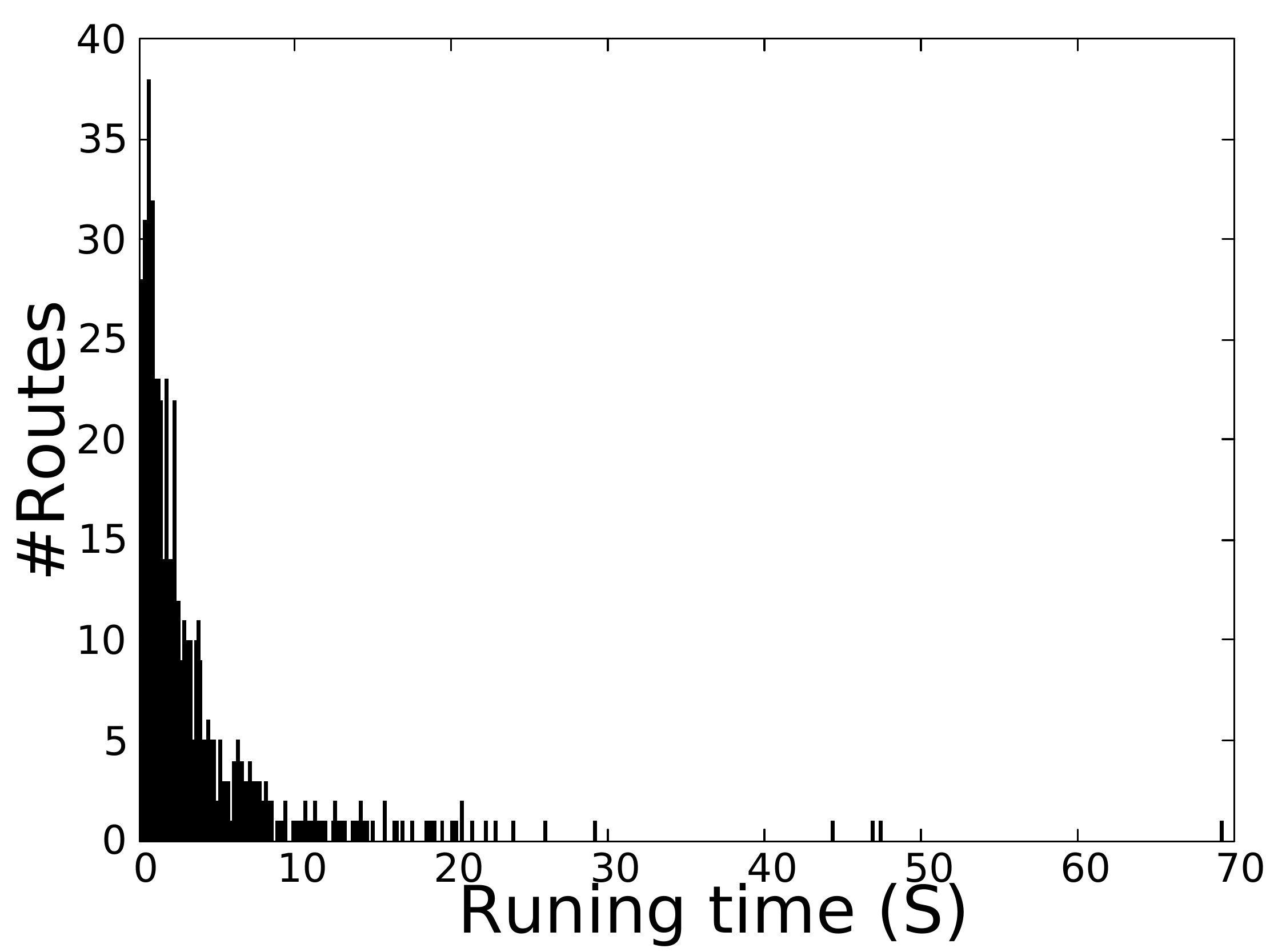}
	}
	\subfigure[\textbf{NYC}]
	{ \label{fig:maxrknnttimenyc}
		\includegraphics[width=4cm]{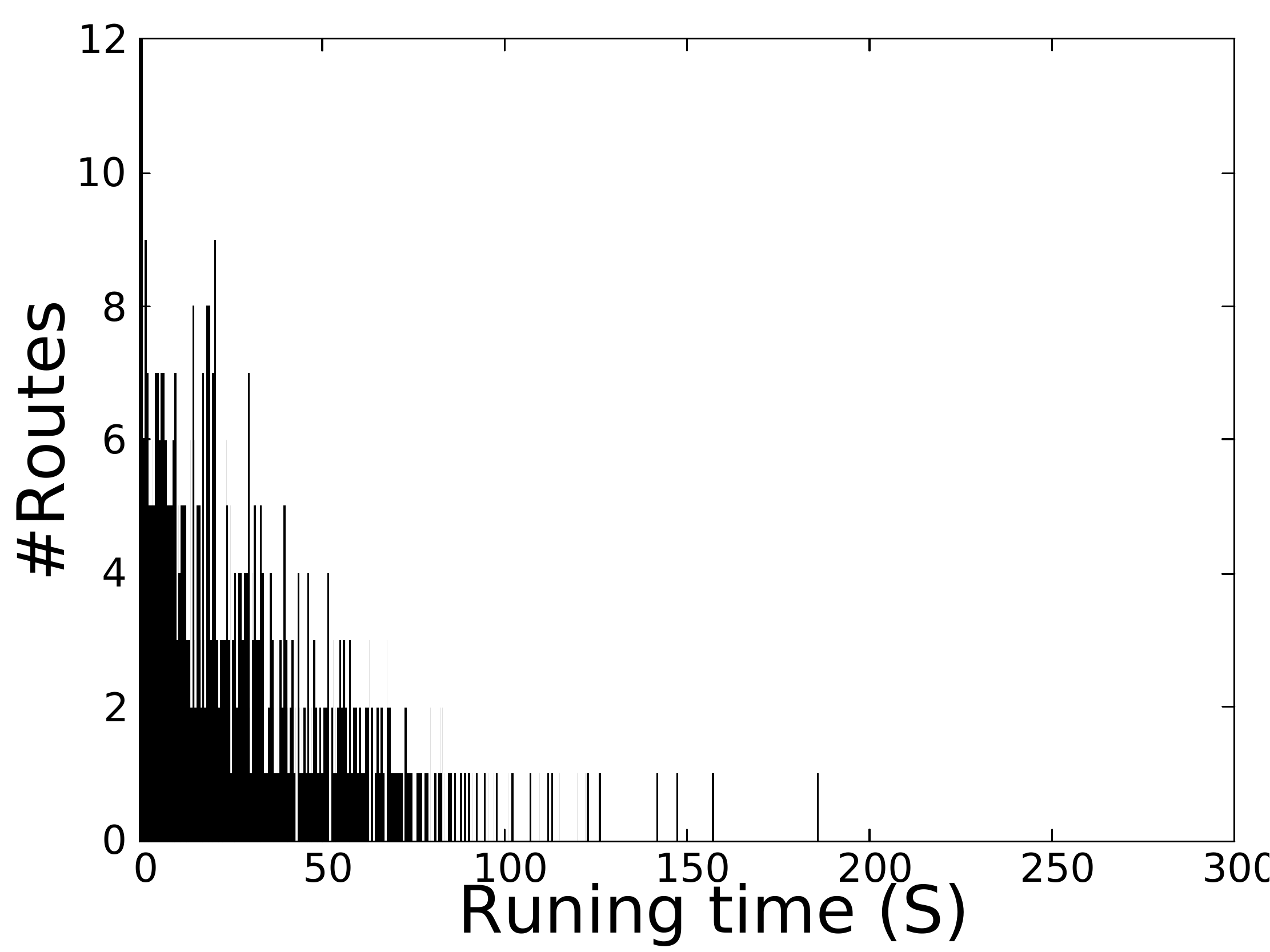}
	}
	\vspace{-2em}
	\caption{Distribution of running time of {\maxrknnt} on real
		route query.}
	\label{fig:maxrknnttime}
	\vspace{-1em}
\end{figure}

\begin{figure}
	\centering
	\begin{minipage}[b]{.5\linewidth}
		\label{fig:fourroutesline}
		\includegraphics[width=4.5cm]{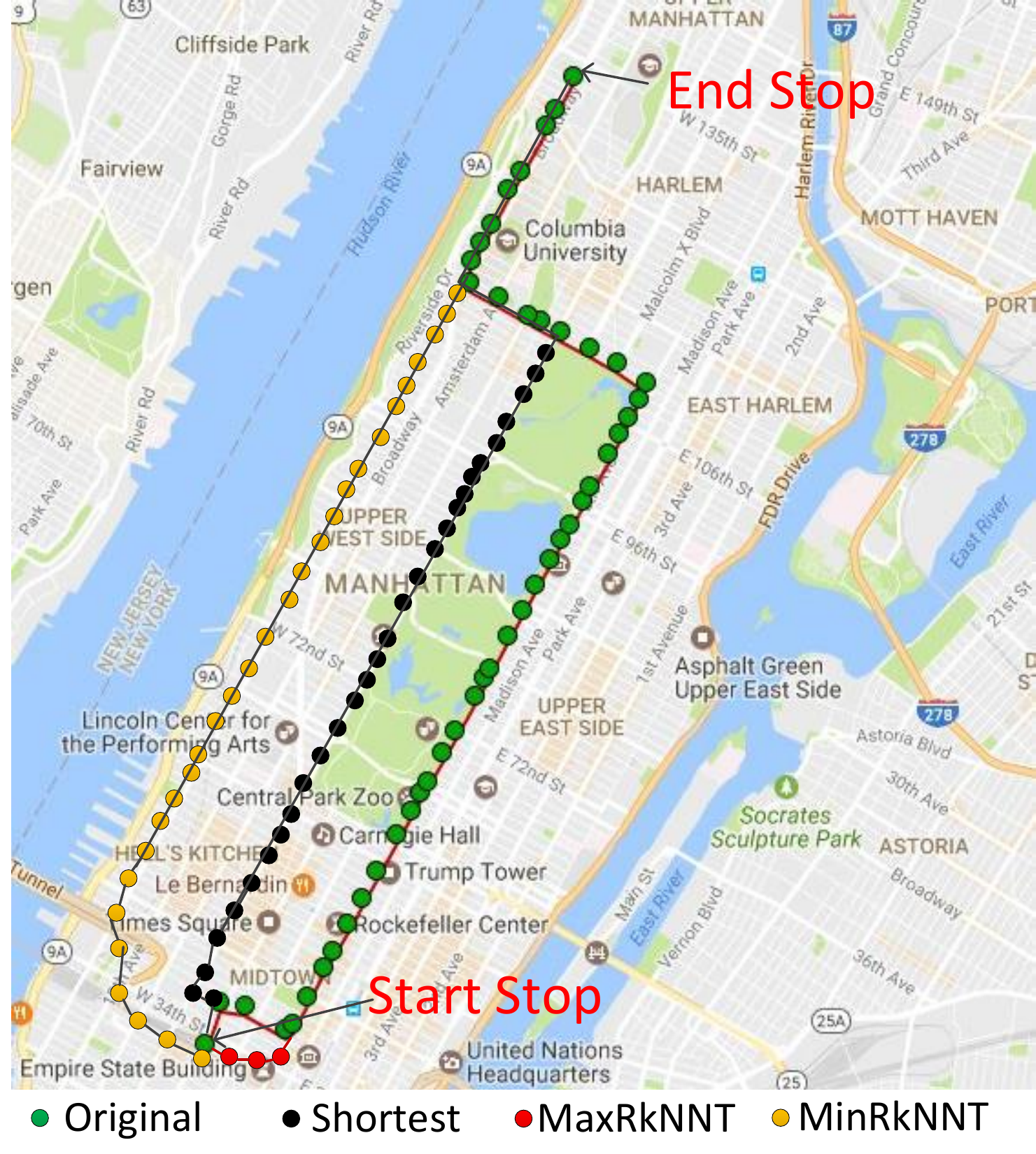}
		\vspace{-1.8em}
	\end{minipage}~~~~~~~~
	\begin{minipage}[b]{.49\linewidth}
		\label{tab:fourroutetable}
		\scomment{
			\begin{tabular}{|c|c|c|c|}
				\hline
				& \textbf{ST} (s) & \textbf{NP} & TD (km)\\ \hline
				1  &  N $\/$ A  &  $1,032$  &  $10,238$  \\ \hline
				2 &  $0.004$  & $817$    &  $9,012$  \\ \hline
				3 &  $1.02$  & $1,061$   &  $10,248$  \\ \hline
				4 &  $0.31$  &  $713$  &   $9,543$ \\ \hline
			\end{tabular}}
			\begin{tabular}{|c|c|c|}
				\hline
				\multicolumn{1}{|c|}{}  & \multicolumn{1}{c|}{\textbf{ST} (s)}  & \multicolumn{1}{c|}{\textbf{NP}}      \\ \hline
				\multicolumn{1}{|c|}{1} & \multicolumn{1}{c|}{N/A}      & \multicolumn{1}{c|}{$1,032$} \\ \hline
				\multicolumn{1}{|c|}{2} & \multicolumn{1}{c|}{$0.004$} & \multicolumn{1}{c|}{$817$}   \\ \hline
				\multicolumn{1}{|c|}{3} & \multicolumn{1}{c|}{$1.02$}  & \multicolumn{1}{c|}{$1,161$} \\ \hline
				\multicolumn{1}{|c|}{4} & \multicolumn{1}{c|}{$0.31$}  & \multicolumn{1}{c|}{$713$}   \\ \hline
				& \textbf{TD} (m)                      & \textbf{\#Stops}                      \\ \hline
				1                       & $10,238$                     &          $49$                    \\ \hline
				2                       & $9,012$                      &            $43$                  \\ \hline
				3                       & $10,248$                     &              $49$                \\ \hline
				4                       & $9,543$                      &               $40$               \\ \hline
			\end{tabular}
			\vspace{2em}
		\end{minipage}
		\caption{Comparison among four routes: ST (searching time), 
			NP (number of passengers), TD (travel distance) and the number of stops.}
		\label{fig:fourroutes}
		\vspace{-1em}
	\end{figure}
	
	In Figure~\ref{fig:fourroutes}, we show four kinds of routes which
	share the same start and end locations: 1) the original bus route passes
	through Manhattan, 2) the shortest distance route,
	3) the {\maxrknnt} route which attracts the most passengers, 4) the
	{\minrknnt} route which attracts the fewest passengers.
	The right table shows the search time, number of passengers,
	travel distance, and number of stops for these four routes.
	We find: (1) the original route and the {\maxrknnt} route are almost
	the same (in particular, {\maxrknnt} finds a route which just is $10$
	meters longer but can attract $129$ extra passengers), which means that
	the existing bus route is almost optimal between the start and
	end locations. 
	(2) If a driver wants to save time, the least crowded route can be
	selected as provided by {\minrknnt}; if the car should be shared
	to increase revenue, the route found by {\maxrknnt} is a good choice.

\section{Conclusion}
\label{sec:conclusion}
In this paper, we proposed and studied the {\rknnt} query, which can
be used directly to support capacity estimation in bus
networks.
First, we proposed a filter-refine processing framework, and an
optimization to increase the filtering space that improves
pruning efficiency.
Then we employed {\rknnt} to solve the bus route planning
problem.
In a bus network, given a start and end bus stop, we can find an optimal
route which attracts the most passengers for a given travel distance
threshold.
To the best of our knowledge, this is the first work studying reverse
$k$ nearest neighbors in trajectories, and our solution
supports dynamically changing transition data while providing
up-to-date answers efficiently.

\begin{small}
\bibliographystyle{abbrv}
\bibliography{sigproc,library}
\end{small}

\end{document}